\documentclass{article}
\usepackage{fullpage}
\usepackage{graphicx}
\usepackage{amsmath}
\usepackage{amssymb}
\usepackage{amsthm}
\usepackage{enumitem}
\usepackage{xcolor}
\usepackage[colorlinks=true]{hyperref}
\usepackage{cleveref}
\usepackage{bbm}
\usepackage{blkarray}
\usepackage{tikz}
\usepackage{authblk}

\newtheorem{theorem}{Theorem}[section]
\crefname{theorem}{Theorem}{Theorems}
\newtheorem{definition}[theorem]{Definition}
\crefname{definition}{Definition}{Definitions}
\newtheorem{lemma}[theorem]{Lemma}
\crefname{lemma}{Lemma}{Lemmas}
\newtheorem{corollary}[theorem]{Corollary}
\crefname{corollary}{Corollary}{Corollaries}

\crefname{figure}{Figure}{Figures}
\crefname{section}{Section}{Sections}
\crefname{subsection}{Subsection}{Subsections}

\counterwithin{figure}{section}

\title{Quantum walks through generalized graph composition}
\author[1]{Arjan Cornelissen}
\affil[1]{Simons institute, UC Berkeley, California, USA}

\newcommand{\A}{\ensuremath{\mathcal{A}}}
\newcommand{\C}{\ensuremath{\mathbb{C}}}
\newcommand{\D}{\ensuremath{\mathcal{D}}}
\newcommand{\E}{\ensuremath{\mathbb{E}}}
\renewcommand{\H}{\ensuremath{\mathcal{H}}}
\newcommand{\K}{\ensuremath{\mathcal{K}}}
\newcommand{\N}{\ensuremath{\mathbb{N}}}
\renewcommand{\P}{\ensuremath{\mathbb{P}}}
\newcommand{\R}{\ensuremath{\mathbb{R}}}
\newcommand{\V}{\ensuremath{\mathcal{V}}}

\newcommand{\ket}[1]{\ensuremath{\left|#1\right\rangle}}
\newcommand{\bra}[1]{\ensuremath{\left\langle#1\right|}}
\newcommand{\braket}[2]{\ensuremath{\left\langle#1\middle|#2\right\rangle}}
\newcommand{\norm}[1]{\ensuremath{\left\|#1\right\|}}
\newcommand{\transduce}{\rotatebox[origin=c]{-90}{$\rightsquigarrow$}}

\DeclareMathOperator{\ADV}{ADV}
\DeclareMathOperator{\diag}{diag}

\DeclareMathOperator{\polylog}{polylog}
\DeclareMathOperator{\Span}{Span}
\DeclareMathOperator{\supp}{supp}
\begin{document}
    \maketitle

    \begin{abstract}
        In this work, we generalize the recently-introduced graph composition framework to the non-boolean setting. A quantum algorithm in this framework is represented by a hypergraph, where each hyperedge is adjacent to multiple vertices. The input and output to the quantum algorithm is represented by a set of boundary vertices, and the hyperedges act like switches, connecting the input vertex to the output that the algorithm computes.

        Apart from generalizing the graph composition framework, our new proposed framework unifies the quantum divide and conquer framework, the decision-tree framework, and the unified quantum walk search framework. For the decision trees, we additionally construct a quantum algorithm from an improved weighting scheme in the non-boolean case. For quantum walk search, we show how our techniques naturally allow for amortization of the subroutines' costs. Previous work showed how one can speed up ``detection'' of marked vertices by amortizing the costs of the quantum walk. In this work, we extend these results to the setting of ``finding'' such marked vertices, albeit in some restricted settings.

        Along the way, we provide a novel analysis of irreducible, reversible Markov processes, by linear-algebraically connecting its effective resistance to the random walk operator. This significantly simplifies the algorithmic implementation of the quantum walk search algorithm, achieves an amortization speed-up for quantum walks over Johnson graphs, avoids the need for quantum fast-forwarding, and removes the log-factors from the query complexity statements.
    \end{abstract}

    \section{Introduction}

    Over the last three decades, the search for efficient quantum algorithms for solving computational problems has led to a wide variety of algorithmic techniques. Notable examples include quantum Fourier sampling~\cite{shor1999polynomial,bernstein1997quantum,simon1997power}, amplitude amplification and estimation~\cite{grover1996fast,brassard2002quantum}, quantum signal processing~\cite{low2017optimal,gilyen2019quantum}, quantum walks~\cite{szegedy2004quantum,magniez2011search,belovs2013quantum,dohotaru2017controlled,apers2021unified}, etc.

    In this search, the number of queries to the computational problem's input has often been used as a measure of efficiency for the designed algorithm. The \textit{query complexity} of a computational problem is the minimum number of times we need to access that problem's input to solve it. This number depends on the computational model, e.g., if we are allowed to use randomness to decide which part of the input to access, we refer to this measure as the \textit{randomized query complexity}. Similarly, if we are allowed to query our input coherently in superposition, then we refer to this quantity as the \textit{quantum query complexity}.

    In a landmark result, Reichardt showed that the quantum query complexity is captured by a semi-definite program, called the \textit{quantum adversary bound}~\cite{reichardt2009span,reichardt2011reflections}. This sets quantum query complexity apart from its randomized counterpart, where no such characterization is known. Moreover, there is a constructive way to turn any feasible solution of the minimization version of the adversary bound into a quantum query algorithm, which was recently simplified considerably using transducers~\cite{belovs2023one,belovs2024taming}. Many algorithmic techniques that followed essentially boil down to constructing feasible solutions to the adversary bound in a smart way. Examples include the learning graph framework~\cite{belovs2012learning}, the electric network framework~\cite{belovs2013quantum}, the $st$-connectivity framework~\cite{belovs-reichardt2012span,jeffery2017quantum,jarret2018quantum}, and the quantum divide-and-conquer framework~\cite{childs2025quantum}.

    Recently, the graph composition framework was introduced, which unifies all of these~\cite{cornelissen2025quantum}. An instance of this framework is an undirected graph $G = (V,E)$, with source and sink nodes $s,t \in V$. For any input $x \in \D$ in some finite domain $\D$, we define a subgraph $G(x) = (V,E_x)$, where the presence of the edge $e \in E_x$ can be determined by a subroutine $\mathcal{P}_e$, encoded (without loss of generality) as a span program on $\D$. The framework then constructs a quantum query algorithm that computes whether $s$ and $t$ are connected in the resulting subgraph $G(x)$.

    An inherent limitation of the graph composition framework is that it only naturally uses decision problems, i.e., it can only access subroutines that solve decision problems, and similarly it can merely produce quantum algorithms that solve decision problems. It is worth noting that other quantum algorithmic frameworks work well with non-boolean inputs and outputs too, like the unified quantum walk framework~\cite{apers2021unified}, or the decision-tree framework~\cite{beigi2020quantum,cornelissen2025improved}. This begs the question whether the graph composition framework can be generalized to the non-boolean setting, which we answer affirmatively in this work.

    \subsection{The generalized graph composition framework}

    In this work, we introduce the \textit{generalized graph composition framework}, which extends the graph composition framework to the state-conversion setting.

    \paragraph{State-conversion.} The state-conversion problem is the problem of converting a state $\ket{\sigma_x}$ into a state $\ket{\tau_x}$, given access to a unitary oracle $O_x$, for some unknown input $x \in \D$ from a finite domain $\D$. This problem encapsulates function evaluation, where $\ket{\sigma_x} := \ket{\perp}$ and $\ket{\tau_x} := \ket{f(x)}$, for every $x \in \D$ and some arbitrary reference state $\ket{\perp}$. Furthermore, it encapsulates ``database updates'', where we want to update a classical database from a state $S_x$ to $T_x$, given an oracle $O_x$. In this case we set $\ket{\sigma_x} := \ket{S_x}$ and $\ket{\tau_x} := \ket{T_x}$ to the computational basis states representing the database's state before and after the updates, respectively.

    In \cite{lee2011quantum}, the adversary bound $\ADV(P)$ for the state-conversion problem $P = \{(\ket{\sigma_x}, \ket{\tau_x}, O_x)\}_{x \in \D}$ is defined as the optimal value of the following semi-definite program:
    \begin{align*}
        \ADV(P) := \min\quad & \max_{x \in \D} \norm{\ket{w_x}}^2, \\
        \text{s.t.}\quad & \bra{w_x}(I - (O_x^{\dagger}O_y) \otimes I_{\mathcal{W}})\ket{w_y} = \braket{\sigma_x}{\sigma_y} - \braket{\tau_x}{\tau_y}, & \forall x,y \in \D, \\
        & \ket{w_x} \in \mathcal{M} \otimes \mathcal{W}, & \forall x \in \D,
    \end{align*}
    where for all $x \in \D$, $\ket{\sigma_x},\ket{\tau_x} \in \V$, we $O_x$ is a unitary acting on $\mathcal{M}$, and $\mathcal{W}$ is a Hilbert space of arbitrary (finite) dimension. It was shown in~\cite[Theorem~4.9]{lee2011quantum} that this adversary bound is tight for the query complexity of state conversion, if we require the quantum algorithm to prepare the state $\ket{\tau_x}$ up to constant fidelity. Furthermore, it was shown that any feasible solution to the minimization version of this optimization program can be turned into a quantum algorithm, and Belovs, Jeffery and Yolcu recently proved the existence of a much simpler construction using transducers~\cite{belovs2023one,belovs2024taming}.

    The adversary bound for state conversion serves as the starting point for this work. We first provide a structural lemma, showing that it suffices to look at a restricted set of state-conversion problems, which we dub state-reflection problems. These associate two orthogonal states $\ket{\sigma_x^+}$ and $\ket{\sigma_x^-}$ to every input $x \in \D$, which they map to $\ket{\sigma_x^+}$ and $-\ket{\sigma_x^-}$, and they are given access to an oracle $O_x$ that satisfies $O_x^2 = I$.

    \begin{lemma}[Informal version of \cref{thm:reformulation-state-conversion}]
        Let $P = \{(\ket{\sigma_x}, \ket{\tau_x}, O_x)\}_{x \in \D}$ be a state-conversion problem. For all $x \in \D$, let $\ket{\sigma_x^{\pm}} := (\ket{\sigma_x} \oplus \pm \ket{\tau_x})/\sqrt{2}$, and $\overline{O_x} = (X \otimes I_{\mathcal{M}})(O_x \oplus O_x^{\dagger}) \in \mathcal{L}(\mathcal{M} \oplus \mathcal{M})$\footnote{Here, $X$ is the Pauli-$X$ matrix, which means that that $X \otimes I_{\mathcal{M}}$ swaps the two copies of $\mathcal{M}$ it acts on.}, and consider the \emph{state-reflection problem} $R := \{(\ket{\sigma_x^+},\ket{\sigma_x^-},O_x)\}_{x \in \D}$, representing the state-conversion problem $\{(\ket{\sigma_x^{\pm}}, \pm\ket{\sigma_x^{\pm}}, \overline{O_x})\}_{x \in \D}$. Then, $\ADV(P) = \ADV(R)$, and we can without loss of generality assume that for all $x \in \D$, $\overline{O_x}\ket{w_x^{\pm}} = \pm\ket{w_x^{\pm}}$, for any feasible solution $w = \{\ket{w_x^{\pm}}\}_{x \in \D}$ to $R$.
    \end{lemma}

    Incidentally, this characterization of state-conversion problems in terms of state-reflection problems gives rise to an explicit construction of the transducers used in the algorithm by Belovs, Jeffery and Yolcu~\cite{belovs2024taming}. Thus, we lift their existence proof to a constructive result, in \cref{thm:transducers}.

    The second benefit of considering state-reflection problems instead of state-conversion problems, is that it naturally introduces the concept of positive and negative witnesses. That is, for any feasible solution $w = \{\ket{w_x^{\pm}}\}_{x \in \D}$, we refer to $\ket{w_x^+}$ as the positive witness for $x$, and $\ket{w_x^-}$ as the negative witness for $x$. This way of looking at witnesses is a proper generalization of the notion of witness vectors that can be found elsewhere in the literature, e.g., in~\cite{reichardt2012span,ito2019approximate}. The size of these witness vectors determine the value of the objective function, and so we write $\mathsf{R}_x^+(w) := \norm{\ket{w_x^+}}^2$ and $\mathsf{R}_x^-(w) := \norm{\ket{w_x^-}}^2$, for every input $x \in \D$.

    Even though state-reflection problems are restrictions of state-conversion problems, there is still a somewhat surprising amount of flexibility within these objects. Indeed, we can freely apply invertible linear operators to $\ket{\sigma_x^{\pm}}$, and rescale the positive and negative witnesses, as described in the following lemma.

    \begin{lemma}[Informal version of \cref{lem:rescale-state-reflection-problem}]
        Let $D \in \mathcal{L}(\mathcal{V})$ be an invertible operator, and $\alpha_{\pm} \in \C$. Then, any feasible solution $w := \{\ket{w_x^{\pm}}\}_{x \in \D}$ for a state-reflection problem $R := \{(\ket{\sigma_x^+}, \ket{\sigma_x^-}, O_x)\}_{x \in \D}$ can be turned into a feasible solution $w' := \{\alpha_+\ket{w_x^+},\alpha_-\ket{w_x^-}\}_{x \in \D}$ for $R' := \{(\alpha_+D\ket{\sigma_x^+}, \alpha_-(D^{-1})^{\dagger}\ket{\sigma_x^{\pm}}, O_x)\}_{x \in \D}$. Consequently,
        \[\ADV(R) \leq \sqrt{\max_{x \in \D} \mathsf{R}_x^+(w) \cdot \max_{x \in \D} \mathsf{R}_x^-(w)}.\]
    \end{lemma}

    \paragraph{The framework.} The fundamental idea of the generalized graph composition framework is to use state-reflections problems of a specific form, such that its states can be interpreted as graph-theoretic objects.

    To that end, let $e$ be a hyperedge, with vertices $V := N(e)$. For every input $x \in \D$, we associate vectors $\delta_x^e, U_x^e \in \C^V$ to the hyperedge $e$, such that the sum of the entries of $\delta_x^e$ vanish and $U_x^e$ is constant on $\supp(\delta_x^e)$. We say that the state-reflection problem $R = \{(\delta_x^e, U_x^e, O_x)\}_{x \in \D}$ implements the ``net-flow'' $\delta_x^e$ and ``potential function'' $U_x^e$ on hyperedge $e$, for input $x \in \D$, as displayed in \cref{fig:hyperedge}. We aptly refer to these as \textit{hyperedge problems}.

    \begin{figure}[!ht]
        \centering
        \begin{tikzpicture}[yscale = -1]
            \node[draw, rounded corners = .3em] (P) at (0,0) {$e$};
            \draw ({cos(0)},{sin(0)}) to (P);
            \fill ({cos(0)},{sin(0)}) circle[radius = .2em];
            \draw[gray!60] ({cos(40)},{sin(40)}) to (P);
            \fill[gray!40] ({cos(40)},{sin(40)}) circle[radius = .2em];
            \draw[gray!60] ({cos(80)},{sin(80)}) to (P);
            \fill[gray!40] ({cos(80)},{sin(80)}) circle[radius = .2em];
            \draw[gray!60] ({cos(120)},{sin(120)}) to (P);
            \fill[gray!40] ({cos(120)},{sin(120)}) circle[radius = .2em];
            \draw ({cos(160)},{sin(160)}) to (P);
            \fill ({cos(160)},{sin(160)}) circle[radius = .2em];
            \draw ({cos(200)},{sin(200)}) to (P);
            \fill ({cos(200)},{sin(200)}) circle[radius = .2em];
            \draw[gray!60] ({cos(240)},{sin(240)}) to (P);
            \fill[gray!40] ({cos(240)},{sin(240)}) circle[radius = .2em];
            \draw[gray!60] ({cos(280)},{sin(280)}) to (P);
            \fill[gray!40] ({cos(280)},{sin(280)}) circle[radius = .2em];
            \draw[gray!60] ({cos(320)},{sin(320)}) to (P);
            \fill[gray!40] ({cos(320)},{sin(320)}) circle[radius = .2em];

            \draw[->] ({1.5*cos(160)},{1.5*sin(160)}) node[left] {$(\delta_x^e)_v$} to ({1.1*cos(160)},{1.1*sin(160)}) node[below] {$v$};
            \draw[->] ({1.5*cos(200)},{1.5*sin(200)}) node[left] {$(\delta_x^e)_{v'}$} to ({1.1*cos(200)},{1.1*sin(200)}) node[below] {$v'$};
            \draw[<-] ({1.5*cos(0)},{1.5*sin(0)}) node[right] {$(\delta_x^e)_{v''}$} to ({1.1*cos(0)},{1.1*sin(0)}) node[below] {$v''$};
        \end{tikzpicture}
        \caption{Pictorial representation of a state-reflection problem implementing a net-flow $\delta_x^e$ and potential function $U_x^e$, for all $x \in \D$, onto the hyperedge $e$. The black vertices represent $\supp(\delta_x^e) \subseteq N(e)$, and so the potential function $U_x^e$ has to be constant on this set, but can be different on the gray vertices. We can think of this state-reflection program $R$ providing a connection between the vertices $v$, $v'$ and $v''$ (the black edges), whilst ``cutting'' all the connections from this hyperedge towards the other vertices (the gray edges).}
        \label{fig:hyperedge}
    \end{figure}
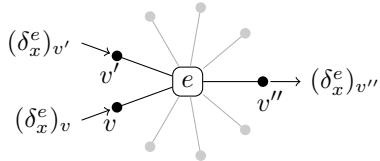

    The idea is that we can construct feasible solutions for hyperedge problems $R = \{(\delta_x^e,U_x^e,O_x)\}_{x \in \D}$ from those for state-conversion problems by choosing suitable invertible linear operators $D$. This is not possible in general, i.e., we cannot always find an invertible linear operator $D$ that maps the states from a state-conversion problem into the form we require. However, we \textit{can} always do this for function evaluation, as we can set $\ket{\sigma_x} := \ket{\perp}$, and $\ket{\tau_x} := \ket{f(x)}$, and let $D$ act as identity on the input states, and $-I$ on the output states. Similarly for database updates, we can set $\ket{\sigma_x} = \ket{S_x}$ and $\ket{\tau_x} = \ket{T_x}$, with the same choice for $D$. In these cases, we can interpret the flows and potential functions as a switch, providing a link between input and output vertices, as described in \cref{fig:function-evaluation-database-updates}.

    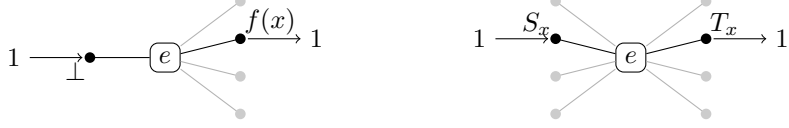
\begin{figure}[!ht]
        \centering
        \begin{tikzpicture}[vertex/.style = {draw, rounded corners = .3em}]
            \node[vertex] (e) at (0,0) {$e$};
            \draw (e) to (-1,0) node[below left=-.2em] {$\perp$};
            \fill (-1,0) circle[radius=.2em];
            \draw[->] (-1.8,0) node[left] {$1$} to (-1.1,0);

            \draw[gray!60] (e) to (1,.75);
            \fill[gray!40] (1,.75) circle[radius=.2em];
            \draw (e) to (1,.25) node[above right=-.2em] {$f(x)$};
            \fill (1,.25) circle[radius=.2em];
            \draw[->] (1.1,.25) to (1.8,.25) node[right] {$1$};
            \draw[gray!60] (e) to (1,-.25);
            \fill[gray!40] (1,-.25) circle[radius=.2em];
            \draw[gray!60] (e) to (1,-.75);
            \fill[gray!40] (1,-.75) circle[radius=.2em];
        \end{tikzpicture}\hspace{5em}\begin{tikzpicture}[vertex/.style = {draw, rounded corners = .3em}]
            \node[vertex] (e) at (0,0) {$e$};
            \draw[gray!60] (e) to (-1,.75);
            \fill[gray!40] (-1,.75) circle[radius=.2em];
            \draw (e) to (-1,.25) node[above left=-.2em] {$S_x$};
            \fill (-1,.25) circle[radius=.2em];
            \draw[<-] (-1.1,.25) to (-1.8,.25) node[left] {$1$};
            \draw[gray!60] (e) to (-1,-.25);
            \fill[gray!40] (-1,-.25) circle[radius=.2em];
            \draw[gray!60] (e) to (-1,-.75);
            \fill[gray!40] (-1,-.75) circle[radius=.2em];

            \draw[gray!60] (e) to (1,.75);
            \fill[gray!40] (1,.75) circle[radius=.2em];
            \draw (e) to (1,.25) node[above right=-.2em] {$T_x$};
            \fill (1,.25) circle[radius=.2em];
            \draw[->] (1.1,.25) to (1.8,.25) node[right] {$1$};
            \draw[gray!60] (e) to (1,-.25);
            \fill[gray!40] (1,-.25) circle[radius=.2em];
            \draw[gray!60] (e) to (1,-.75);
            \fill[gray!40] (1,-.75) circle[radius=.2em];
        \end{tikzpicture}
        \caption{Flow and potential function for function evaluation (left) and database update (right). In both cases, the potential function is $1$ and $0$ on all the black and gray vertices, respectively. Intuitively, we can think of the state-reflection problem implementing a switch between an input vertex on the left, and an output vertex on the right.}
        \label{fig:function-evaluation-database-updates}
    \end{figure}

    Another way in which we can obtain feasible solutions for hyperedge problems is from span programs. Indeed, for any span program $\mathcal{P}$ on $\D$, we take a simple edge $e = vw$. This edge now acts as a switch that can be either open, when $x$ is a negative input for $\mathcal{P}$, or closed otherwise. If the switch is closed, a unit flow can pass through the edge, whereas when it is open, the switch exhibits a unit potential difference. The witnesses from the span program framework now exactly carry over to feasible solutions to the corresponding hyperedge problem, as proved in \cref{thm:hyperedge-from-span-program}. We provide a pictorial representation of the resulting hyperedge problem in \cref{fig:span-programs}.

    \begin{figure}[!ht]
        \centering
        \begin{tikzpicture}[vertex/.style = {draw, rounded corners = .3em}]
            \node[vertex] (P) at (1,0) {$e$};
            \draw (0,0) node[below] {$v$} to (P) to (2,0) node[below] {$w$};
            \fill (0,0) circle[radius = .2em];
            \fill (2,0) circle[radius = .2em];
            \draw[->] (-.5,0) node[left] {$1$} to (-.1,0);
            \draw[<-] (2.5,0) node[right] {$1$} to (2.1,0);
        \end{tikzpicture}\hspace{3em}\begin{tikzpicture}[vertex/.style = {draw, rounded corners = .3em}]
            \node[vertex] (P) at (1,0) {$e$};
            \draw[gray!60] (0,0) node[below] {$v$} to (P) to (2,0) node[below] {$w$};
            \fill[gray!40] (0,0) circle[radius = .2em];
            \fill[gray!40] (2,0) circle[radius = .2em];
        \end{tikzpicture}
        \caption{The hyperedge problem implemented by a span program. On the left-hand side, we have a positive input that implements a unit flow through the edge $e$. This means that the potential function has to be constant on both sides of the edge. On the right-hand side, we have a negative input blocking any flow going through the edge, and thereby allowing there to be a potential difference across the edge.}
        \label{fig:span-programs}
    \end{figure}
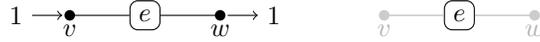

    An instance of the generalized graph composition framework consists of the following ingredients:
    \begin{enumerate}[nosep]
        \item A hypergraph $G = (V,E)$, with weights $w_e > 0$ for all $e \in E$, and a set of boundary vertices $B \subseteq V$, and internal vertices $V \setminus B$.
        \item A hyperedge problem $R^e = \{(\delta_x^e, U_x^e, O_x)\}_{x \in \D}$ for every edge $e \in E$, such that
        \begin{enumerate}[nosep]
            \item The net-flow vanishes at every internal vertex, i.e., $\sum_{e \in N(v)} \delta_x^e(v) = 0$, where $N(v) \subseteq E$ is the set of hyperedges adjacent to $v \in V \setminus B$.
            \item The potential function is uniquely defined at every vertex, i.e., $U_x^e(v)$ is the same for all hyperedges $e$ adjacent to $v$.
        \end{enumerate}
    \end{enumerate}
    For all $v \in B$, we define the net-flow $\delta_x(v) = \sum_{e \in N(v)} \delta_x^e(v)$, and similarly we let $U_x$ be the potential function on the boundary vertices. The framework combines feasible solutions to all the hyperedge problems $R^e$ with complexities $(\mathsf{R}^e)_x^{\pm}(w^e)$ into a feasible solution to the state-reflection problem $\{(\delta_x,U_x,O_x)\}_{x \in \D}$, with complexities $\mathsf{R}_x^{\pm}(w) = \sum_{e \in E} w_e^{\pm1}(\mathsf{R}^e)_x^{\pm}(w^e)$. For the details, we refer to \cref{thm:generalized-graph-composition}.

    \paragraph{Resistance-cut theorem.} In its general form, the generalized graph composition framework can be somewhat cumbersome to use. Therefore, we derive a simpler but less powerful version, akin to \cite[Theorem~4.3]{cornelissen2025quantum}.

    \begin{theorem}[Informal version of the resistance-cut theorem, \cref{thm:path-cut-theorem}]
        Let $G = (V,E)$ be a hypergraph, with boundary vertices $B \subseteq V$. For every hyperedge $e \in E$, let $R^e$ be a hyperedge problem with for all inputs $x \in \D$ a unit flow between two vertices $\delta_x^e = \mathbbm{1}_v - \mathbbm{1}_w$, and corresponding potential function $U_x^e = \mathbbm{1}_v + \mathbbm{1}_w$. Let $G(x) = (V,E(x))$ be the connected undirected graph that contains all these connections, and let $B_x \subseteq B$ be the subset of boundary vertices in the connected component of $G(x)$. Let $\delta_x$ be a net-flow on the boundary vertices with support on $B_x$. On the other hand, let $C_x \subseteq E$ be a set of hyperedges that cuts the connected component of $G(x)$ from the other boundary vertices $V \setminus B_x$. Then, we can find a feasible solution to the state-reflection problem $\{(\delta_x, \mathbbm{1}_{B_x}, O_x)\}_{x \in \D}$, with complexities
        \[\mathsf{R}_x^+ = R_{\mathrm{eff}}(G(x), (w_e \cdot (\mathsf{R}^e)_x^+)_{e \in E(x)}; \delta_x)\footnote{Here, we use a slightly more general definition for the effective resistance. See \cref{def:effective-resistance}.}, \qquad \text{and} \qquad \mathsf{R}_x^- = \sum_{e \in C_x} \frac{(\mathsf{R}^e)_x^-}{w_e}.\]
    \end{theorem}

    We show several direct applications of this theorem in \cref{subsec:examples}.

    \subsection{Unification results}

    The first contribution we make in this work is a unification of several algorithmic frameworks. We first formally provide the (non-surprising) result that the  generalized graph composition framework indeed generalizes the graph composition framework.

    \begin{theorem}[Informal version of \cref{thm:graph-composition}]
        Any instance of the graph composition framework from \cite{cornelissen2025quantum} is also an instance to generalized graph composition.
    \end{theorem}

    Next, we show that generalized graph compositions encapsulate instances of the decision tree frameworks that are introduced by \cite{beigi2020quantum} and \cite{cornelissen2025improved}.

    \begin{theorem}[Informal version of \cref{thm:BT20,thm:CMP25}]
        Every instance of the decision-tree frameworks of Beigi and Taghavi~\cite{beigi2020quantum} and Cornelissen, Mande and Patro~\cite{cornelissen2025improved}, is also an instance of the generalized graph composition framework.
    \end{theorem}

    This unification not only leads to a new interpretation of these aforementioned techniques, but it also allows for generalizing the weighting scheme introduced in \cite{cornelissen2025improved} to the setting of non-boolean inputs, as considered in \cite{beigi2020quantum}. Whether this leads to significant improvements for specific applications remains to be seen, we leave this for future research.

    Next, we show how the generalized graph composition framework encapsulates the quantum divide-and-conquer framework. We recall that the simple graph composition framework already recovers Strategy~1 from \cite{childs2025quantum}, and we show here that we can also recover Strategy~2.

    \begin{theorem}[Informal version of \cref{thm:divide-and-conquer}]
        Any instance of Strategy~2 of the divide-and-conquer framework from \cite{childs2025quantum} is also an instance of the generalized graph composition framework.
    \end{theorem}

    Finally, we show that any instance of the unified quantum walk search framework, introduced in \cite{apers2021unified} is also an to the generalized graph composition framework. Our result holds in a slightly different setting, though, which we elaborate on in \cref{subsec:quantum-walk-search}.

    \begin{theorem}[Informal version of \cref{thm:quantum-walk-detection}]
        Any instance of the detection version of quantum walk search from \cite{apers2021unified}, is also an instance of the generalized graph composition framework.
    \end{theorem}

    Merging all these results with the observations from \cite{cornelissen2025quantum}, we obtain an overview of the relations between quantum algorithmic frameworks that is displayed in \cref{fig:overview-frameworks}.

    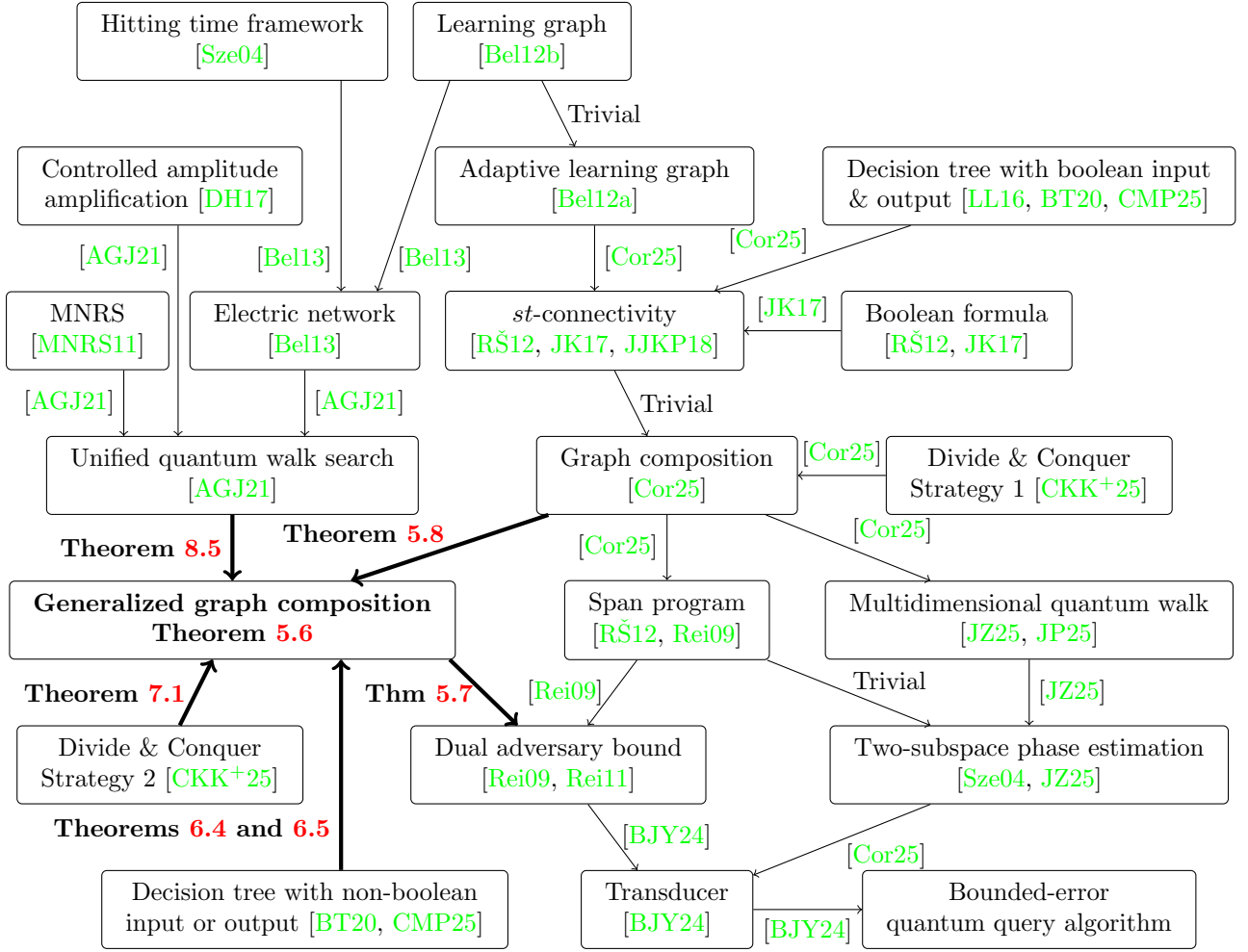
\begin{figure}[!ht]
        \centering
        \begin{tikzpicture}[vertex/.style = {draw, rounded corners = .2em}]
            \node[vertex] (NBDT) at (-5,-4) {\begin{tabular}{c}
                Decision tree with non-boolean \\
                input or output \cite{beigi2020quantum,cornelissen2025improved}
            \end{tabular}};
            \node[vertex] (UQW) at (-6,2) {\begin{tabular}{c}
                Unified quantum walk search \\
                \cite{apers2021unified}
            \end{tabular}};
            \node[vertex] (MNRS) at (-8,4) {\begin{tabular}{c}
                MNRS \\
                \cite{magniez2011search}
            \end{tabular}};
            \node[vertex] (CQA) at (-7,6) {\begin{tabular}{c}
                Controlled amplitude \\
                amplification \cite{dohotaru2017controlled}
            \end{tabular}};
            \node[vertex] (EN) at (-5,4) {\begin{tabular}{c}
                Electric network \\
                \cite{belovs2013quantum}
            \end{tabular}};
            \node[vertex] (HT) at (-6,8) {\begin{tabular}{c}
                Hitting time framework \\
                \cite{szegedy2004quantum}
            \end{tabular}};

            \node[vertex] (DC2) at (-7,-2) {\begin{tabular}{c}
                Divide \& Conquer \\
                Strategy 2 \cite{childs2025quantum}
            \end{tabular}};
            \node[vertex] (GGC) at (-6,0) {\begin{tabular}{c}
                \bf Generalized graph composition \\
                \bf\cref{def:generalized-graph-composition}
            \end{tabular}};
            \node[vertex] (GC) at (0,2) {\begin{tabular}{c}
                Graph composition \\
                \cite{cornelissen2025quantum}
            \end{tabular}};
            \node[vertex] (SP) at (0,0) {\begin{tabular}{c}
                Span program \\
                \cite{reichardt2012span,reichardt2009span}
            \end{tabular}};
            \node[vertex] (st) at (-1,4) {\begin{tabular}{c}
                $st$-connectivity \\
                \cite{reichardt2012span,jeffery2017quantum,jarret2018quantum}
            \end{tabular}};
            \node[vertex] (ALG) at (-1,6) {\begin{tabular}{c}
                Adaptive learning graph \\
                \cite{belovs2012learning}
            \end{tabular}};
            \node[vertex] (LG) at (-2,8) {\begin{tabular}{c}
                Learning graph \\
                \cite{belovs2012span}
            \end{tabular}};

            \node[vertex] (DC1) at (5,2) {\begin{tabular}{c}
                Divide \& Conquer \\
                Strategy 1 \cite{childs2025quantum}
            \end{tabular}};
            \node[vertex] (FE) at (4,4) {\begin{tabular}{c}
                Boolean formula \\
                \cite{reichardt2012span,jeffery2017quantum}
            \end{tabular}};
            \node[vertex] (DT) at (5,6) {\begin{tabular}{c}
                Decision tree with boolean input \\
                \& output \cite{lin2016upper,beigi2020quantum,cornelissen2025improved}
            \end{tabular}};

            \node[vertex] (MQW) at (5,0) {\begin{tabular}{c}
                Multidimensional quantum walk \\
                \cite{jeffery2025multidimensional,jeffery-pass2025multidimensional}
            \end{tabular}};
            \node[vertex] (PE) at (5,-2) {\begin{tabular}{c}
                Two-subspace phase estimation \\
                \cite{szegedy2004quantum,jeffery2025multidimensional}
            \end{tabular}};
            \node[vertex] (ADV) at (-1.5,-2) {\begin{tabular}{c}
                Dual adversary bound \\
                \cite{reichardt2009span,reichardt2011reflections}
            \end{tabular}};
            \node[vertex] (T) at (0,-4) {\begin{tabular}{c}
                Transducer \\
                \cite{belovs2024taming}
            \end{tabular}};
            \node[vertex] (Q) at (5,-4) {\begin{tabular}{c}
                Bounded-error \\
                quantum query algorithm
            \end{tabular}};

            \draw[->] ([shift={(1.5,0)}]HT.south) to node[left, pos=.85] {\cite{belovs2013quantum}} ([shift={(.5,0)}]EN.north);
            \draw[->] ([shift={(-1,0)}]LG.south) to node[right, pos=.85] {\cite{belovs2013quantum}} ([shift={(1,0)}]EN.north);
            \draw[->] ([shift={(.25,0)}]CQA.south) to node[pos=.15,left] {\cite{apers2021unified}} ([shift={(-.75,0)}]UQW.north);
            \draw[->] ([shift={(.5,0)}]MNRS.south) to node[left] {\cite{apers2021unified}} ([shift={(-1.5,0)}]UQW.north);
            \draw[->] (EN.south) to node[right] {\cite{apers2021unified}} ([shift={(1,0)}]UQW.north);
            \draw[ultra thick,->] (UQW) to node[left] {\bf\cref{thm:unified-quantum-walk-search}} (GGC);
            \draw[->] (LG) to node[right] {Trivial} (ALG);
            \draw[->] (ALG) to node[right] {\cite{cornelissen2025quantum}} (st);
            \draw[->] (st) to node[right] {Trivial} (GC);
            \draw[->] (GC) to node[left] {\cite{cornelissen2025quantum}} (SP);
            \draw[ultra thick,->] (GC) to node[above left=-.2em] {\bf\cref{thm:graph-composition}} (GGC);
            \draw[->] (DC1) to node[above] {\cite{cornelissen2025quantum}} (GC);
            \draw[->] (FE) to node[above] {\cite{jeffery2017quantum}} (st);
            \draw[->] (DT) to node[above left=-.2em] {\cite{cornelissen2025quantum}} (st);
            \draw[ultra thick,->] ([shift={(.5,0)}]NBDT.north) to node[pos=.2, left] {\bf\cref{thm:BT20,thm:CMP25}} ([shift={(1.5,0)}]GGC.south);
            \draw[ultra thick,->] (DC2) to node[left] {\bf\cref{thm:divide-and-conquer}} (GGC);
            \draw[->] (GC) to node[above right=-.2em] {\cite{cornelissen2025quantum}} (MQW);
            \draw[->] (MQW) to node[right] {\cite{jeffery2025multidimensional}} (PE);
            \draw[->] (PE) to node[below right=-.2em] {\cite{cornelissen2025quantum}} (T);
            \draw[->] (SP) to node[left] {\cite{reichardt2009span}} (ADV);
            \draw[ultra thick,->] ([shift={(3,0)}]GGC.south) to node[left] {\bf Thm~\ref{thm:generalized-graph-composition}} (ADV);
            \draw[->] (ADV) to node[right] {\cite{belovs2024taming}} (T);
            \draw[->] (T) to node[below] {\cite{belovs2024taming}} (Q);
            \draw[->] (SP) to node[above right=-.2em] {Trivial} (PE);
        \end{tikzpicture}
        \caption{Overview of quantum algorithmic frameworks. An arrow from $A$ to $B$ represents that every instance of $A$ can be turned into an instance of $B$. The cited resources represent the first reference where these frameworks are considered in the context of quantum algorithms. The new results proved in this work are represented by the bold arrows.}
        \label{fig:overview-frameworks}
    \end{figure}

    \subsection{Improvements for quantum walk search}
    \label{subsec:quantum-walk-search}

    In the quantum walk search problem, we are given an undirected graph $G = (V,E)$, with resistances $r : E \to \R_{>0}$. For any given input $x \in \D$ where the domain $\D$ is an arbitrary set, we have a set of marked vertices $M_x \subseteq V$. In the \textit{detection} problem, the goal is to determine whether $M_x$ is non-empty. On the other hand, in the \textit{finding} problem, we are promised that $M_x$ is non-empty, and the goal is to output an element from it.

    We can define a Markov process on the vertices of $G$ using a specific random walk dynamic, where at a given vertex $v$, we select an edge $e$ adjacent to $v$ to traverse with probability proportional to $1/r_e$. This setting has been well-studied in the literature and exhibits interesting connections between statistical quantities on the one side, like hitting and commute times, and effective resistances in the graph on the other. We highlight such results in \cref{subsec:random-walks}, insofar as that they're relevant to our work.

    \paragraph{Prior unification result~\cite{apers2021unified}.} In \cite{apers2021unified}, we are given access to the following three subroutines:
    \begin{enumerate}[nosep]
        \item A \textit{setup routine} $S$ with cost $\mathsf{S}$, that generates a superposition with some fixed distribution $\sigma$, i.e.,
        \[S : \ket{0} \mapsto \sum_{v \in V} \sqrt{\sigma_v}\ket{v}.\]
        \item An \textit{update routine} $U$ with cost $\mathsf{U}$, that given $v$ prepares a weighted uniform distribution of the neighbors of $v$, i.e., that acts as
        \[O_P : \ket{v}\ket{0} \mapsto \ket{v} \sum_{w \in N(v)} \sqrt{\frac{r_v}{r_e}}\ket{w}, \qquad \text{where} \qquad r_v = \left[\sum_{e \in N(v)} \frac{1}{r_e}\right]^{-1}.\]
        \item A \textit{checking routine} $C$ with cost $\mathsf{C}$, that given a vertex $v$ checks whether $v \in M_x$.
    \end{enumerate}
    The result in \cite{apers2021unified} is that we can solve the finding problem with high probability with a number of queries that satisfies\footnote{In \cite{apers2021unified}, it is assumed that $R_{\mathrm{eff}}(P^t; \sigma \leftrightarrow M_x) \in \Omega(1)$. However, their results also hold without this assumption, at the expense of an additional additive factor of $\mathsf{C}$, which is the complexity we display here.}
    \begin{align*}
        &\widetilde{O}\left(\mathsf{S} + \max_{x \in \D : M_x \neq \varnothing} \sqrt{R_{\mathrm{eff}}(P^t; \sigma  \leftrightarrow M_x)}(\sqrt{t}\mathsf{U} + \mathsf{C}) + \mathsf{C}\right), & \text{(detection)}, \\
        &\widetilde{O}\left(\mathsf{S} + \sqrt{R_{\mathrm{eff}}(P^t; \sigma  \leftrightarrow M_x)}(\sqrt{t}\mathsf{U} + \mathsf{C}) + \mathsf{C}\right), & \text{(finding)}.
    \end{align*}
    Here, $t \in \N$ is an integer parameter in the construction of the algorithm, and $R_{\mathrm{eff}}(P^t; \sigma \leftrightarrow M_x)$ represents the minimum effective resistance in the graph that is walked on if we take $t$ steps of the random walk $P$, between the initial distribution $\sigma$ and any distribution that has support in $M_x$. This result unifies four previous approaches, \cite{szegedy2004quantum,magniez2011search,belovs2013quantum,dohotaru2017controlled}.

    \paragraph{Amortization by \cite{jeffery2022quantum}.} A property of random walks in the classical setting, is that they are amenable to amortization. That is, if figuring out which edge to traverse is much more costly for some vertices $v \in V$ in the graph, but these vertices are not likely to be visited often, then we can take the ``expected cost per step'' in the algorithm's analysis. This is crucial for, for instance, the best-known algorithm to compute the volume of a convex body~\cite{cousins2018gaussian}. In the quantum setting, on the other hand, we don't know in which vertex we are in any of the intermediate steps, and hence we cannot easily access such amortization techniques.

    Towards achieving such an amortization result, Jeffery developed a variable-time quantum walk search algorithm~\cite{jeffery2022quantum}. They assume that it is free to select the edge that is to be traversed, but the edge traversal itself is costly. Their analysis also works with respect to time complexity, but here we merely state the query complexity results:

    \begin{enumerate}[nosep]
        \item A \textit{setup routine} making $\mathsf{S}$ queries that computes a superposition over a distribution $\sigma$ as before.
        \item An \textit{update routine} that for any $v \in V$ and index $j \in [\deg(v)]$ computes the $j$th neighbor of $v$, say $w$, and the index of $v$ in the list of neighbors of $w$. The number of queries this routine makes is described by the random variable $\mathsf{T}_{vw,x}$.
        \item A \textit{checking routine} that for any $v \in V$ checks whether $v \in M_x$ with a number of queries described by the random variable $\mathsf{C}_{v,x}$.
    \end{enumerate}
    Now, Jeffery showed \cite[Corollary~4.2]{jeffery2022quantum} that we can solve the detection version of quantum walk search in a cost that satisfies
    \begin{align*}
        &\widetilde{O}\left(\mathsf{S} + \sqrt{\max_{x \in \D_+} R_{\mathrm{eff}}(P; \sigma - \nu_x) \cdot \max_{x \in \D_-} \underset{\substack{v \sim \pi \\ w \sim P_{v,\cdot}}}{\E} \left[\mathsf{T}_{vw,x}^2\right]} + \sqrt{\max_{x \in \D_+} \underset{v \sim \tau}{\E} \left[\left(\frac{(\nu_x)_v}{\tau_v}\right)^2\right] \cdot \max_{x \in \D_-} \underset{v \sim \tau}{\E} \left[\mathsf{C}^2_{v,x}\right]}\right),
    \end{align*}
    where $\D_+ = \{x \in \D : M_x \neq \varnothing\}$, $\D_- = \D \setminus \D_+$, $\tau$ is a distribution on $V$ and $\nu_x$ a distribution on $M_x \cap \supp(\tau)$. Moreover, the above complexities can be slightly improved if we know the expected query complexities beforehand.

    Jeffery also specializes these bounds to the specific setting of detection in the MNRS-framework~\cite[Corollary~4.3]{jeffery2022quantum}, with the resulting query complexity being
    \[\widetilde{O}\left(\mathsf{S} + \sqrt{\max_{x \in \D_+} \frac{1}{\underset{v \sim \tau}{\P} [v \in M_x]}} \left(\sqrt{\frac{1}{\delta} \max_{x \in \D_-} \underset{\substack{v \sim \pi \\ w \sim P_{v,\cdot}}}{\E}\left[\mathsf{T}_{vw,x}^2\right]} + \sqrt{\max_{x \in \D_-} \underset{v \sim \tau}{\E} \left[\mathsf{C}^2_{v,x}\right]}\right) \sqrt{\log\frac{1}{\pi_{\min}}}\right),\]
    where $\delta$ is the spectral gap of the graph $G$, and where, again the squares can be moved outside the expectation if the expected query complexity is known beforehand.\footnote{We are sweeping some details under the rug here -- Jeffery shows that one can even move the squares outside the expectation if the costs can be computed on the fly, or with a single query. We refer to their paper for the precise statements.}

    \paragraph{Our improvement.} In this work, we consider again a slightly different setting, where we have a database with possible states $\mathcal{S}$, such that for every vertex $v \in V$ and input $x \in \D$, there exists a unique entry in the database $D_{v,x} \in \mathcal{S}$. The task, now, is to keep track of the database entry during the walk, such that we can use it in the checking routine. This mimics the approach taken in the quantum walk algorithm for the element-distinctness problem~\cite{ambainis2007quantum}.

    We assume that we are given access to the following subroutines, as feasible solutions to their corresponding state-reflection problems:
    \begin{enumerate}[nosep]
        \item For all $v \in V$, a setup routine $S_v$ with witness sizes $(\mathsf{S}_v)_x^{\pm}$, that computes $D_{v,x}$.
        \item For all $vw = e \in E$, an update routine $U_e$ with witness sizes $(\mathsf{U}_e)_x^{\pm}$, that updates $D_{v,x}$ to $D_{w,x}$.
        \item For all $v \in V$, and $D \in \mathcal{S}$, a checking routine $C_{v,D}$, such that for all $x \in \D$, if $D = D_{v,x}$, $C_{v,D}$ computes whether $v \in M_x$, with witness sizes $(\mathsf{C}_v)_x^{\pm}$.
    \end{enumerate}

    We now employ the generalized graph composition construction to obtain the following quantum algorithm that solves the detection version of the quantum walk search problem.

    \begin{theorem}[Simplified version of \cref{thm:quantum-walk-detection}]
        Let $\D_+ = \{x \in \D : M_x \neq \varnothing\}$ and $\D_- = \D \setminus \D_+$. For all $x \in \D^+$, let $\mu_x$ be a probability distribution on $\supp(\sigma)$, and $\nu_x$ be a probability distribution on $M_x \cap \supp(\tau)$. Then, we can solve the detection version of the quantum walk search problem with high probability with a number of queries that satisfies $O(\sqrt{\mathsf{S}^+\mathsf{S}^-} + \sqrt{\mathsf{U}^+\mathsf{U}^-} + \sqrt{\mathsf{C}^+\mathsf{C}^-})$, where
        \begin{align*}
            \mathsf{S}^+ &= \max_{x \in \D_+} \underset{v \sim \sigma}{\E}\left[\left(\frac{(\mu_x)_v}{\sigma_v}\right)^2 (\mathsf{S}_v)_x^+\right], & \mathsf{S}^- &= \max_{x \in \D_-} \underset{v \sim \sigma}{\E}\left[(\mathsf{S}_v)_x^-\right], \\
            \mathsf{U}^+ &= \max_{x \in \D_+} R_{\mathrm{eff}}(P; \mu_x - \nu_x) \cdot \max_{e \in E} (\mathsf{U}_e)_x^+, & \mathsf{U}^- &= \max_{x \in \D_-} \underset{\substack{v \sim \pi \\ w \sim P_{v,\cdot}}}{\E} \left[(\mathsf{U}_{vw})_x^-\right], \\
            \mathsf{C}^+ &= \max_{x \in \D_+} \underset{v \sim \tau}{\E}\left[\left(\frac{(\nu_x)_v}{\tau_v}\right)^2 (\mathsf{C}_{v,D_{v,x}})_x^+\right], & \mathsf{C}^- &= \max_{x \in \D_-} \underset{v \sim \tau}{\E}\left[(\mathsf{C}_{v,D_{v,x}})_x^-\right].
        \end{align*}
    \end{theorem}

    We generalize this result to the finding version in two cases, with very similar complexities. First, in the case where we are promised that the marked vertex is unique (\cref{thm:quantum-walk-finding-unique}), and second, in the case where we know that the fraction of marked vertices is $\varepsilon$ relative to some fixed probability distribution $\tau \in \Delta_V$, that is, $\P_{v \sim \tau}[v \in M_x] = \varepsilon$ (\cref{thm:quantum-walk-finding-fixed-fraction}).

    \paragraph{Analysis of our improvement.} The positive and negative witness complexities from the above theorem statement look quite unwieldy, and it is consequently not easy to see how they relate to the aforementioned query complexities that were obtained by \cite{apers2021unified} and \cite{jeffery2022quantum}. Indeed, in showing that we can recover their results from our complexities, we need to prove some new relations between the effective resistance in a graph $G$ with resistances $r$, and the probability transition matrix $P$ of a random walk on it. We prove that

    \begin{lemma}[Informal version of \cref{lem:fast-forwarding}]
        Let $M = (V,P)$ be an irreducible, reversible Markov process, with stationary distribution $\pi$. Let $\xi \in \R_V$ such that $\sum_{v \in V} \xi_v = 0$. We write $D = \diag(\pi)$, and $\widetilde{\xi} = D^{-1/2}\xi$. Then, we have the following two consequences:
        \begin{enumerate}[nosep]
            \item For integer $t \geq 1$, $R_{\mathrm{eff}}(P;\xi) \leq t \cdot R_{\mathrm{eff}}(P^t;\xi)$.
            \item If the spectral gap of $P$ is $\delta$, we have $R_{\mathrm{eff}}(P;\xi) \leq \|\widetilde{\xi}\|_2^2/\delta$.
        \end{enumerate}
    \end{lemma}

    Both results were already known through combinatorial arguments in the case where $\xi$ is the difference of two standard basis vectors, i.e., $\mathbbm{1}_s - \mathbbm{1}_t$, see e.g., \cite{mieghem2023graph}. However, for our purposes we need the more general case where $\xi$ is the difference between two probability distributions, and so we give a direct linear algebraic proof here. We deem this to be an independently interesting result in its own right.

    \paragraph{Contributions.} Using the aforementioned analysis of the effective resistance, we are now able to recover the complexities from \cite{apers2021unified} and \cite{jeffery2022quantum}. We compare our results to theirs.

    \begin{enumerate}
        \item Encoding the subroutines as feasible solutions to the adversary bound for state-conversion is beneficial over direct algorithmic access models, as used by \cite{apers2021unified} and \cite{jeffery2022quantum}. With algorithmic input models, namely, one can on the one hand consider \textit{zero-error} quantum algorithms, which is inherently limited since some problems can be much harder to solve with zero-error than with bounded error (e.g., Grover search). On the other hand, one can consider \textit{bounded-error} quantum algorithms, but then one has to deal with the build-up of errors throughout the procedure, complicating the construction and likely leading to polylogarithmic overhead of error reduction along the way.

        Feasible solutions to the adversary bound possess both the desirable properties: they are exact in the sense that they solve the adversary bound exactly, and they are as powerful as bounded-error quantum query algorithms, due to the characterization of bounded-error quantum query complexity by the adversary bound~\cite{lee2011quantum}. Developments after \cite{jeffery2022quantum} suggest how to use Jeffery's techniques without build-up of errors using Las-Vegas algorithms~\cite{belovs2023one}, which we recover with our approach, see~\cref{thm:las-vegas-algorithm}.

        \item We recover the same query complexity as in \cite{apers2021unified} for the detection version of quantum walk search, but without any log-factors. See \cref{thm:unified-quantum-walk-search}. We obtain the same improved result for the finding case, if we are promised that there is only one marked vertex. We conjecture this result can also be obtained in the general finding case, but we leave that to future research.

        \item An interesting consequence of our improved analysis is that we recover the query bounds from \cite{apers2021unified}, for all values of the parameter $t$. However, the construction of the quantum algorithm we propose does not depend on $t$, i.e., it is merely an artifact of the analysis. Since the parameter $t$ is used in the quantum fast-forwarding step of \cite{apers2021unified}, we remove the necessity of doing quantum fast-forwarding for quantum walk search, conceptually simplifying its implementation.

        \item We recover all the complexities from \cite{jeffery2022quantum}. With our improved analysis of the effective resistance, we can additionally improve the query complexity of the variable-time MNRS framework, by removing the additional factor $\sqrt{\log(1/\pi_{\min})}$. This was already conjectured to be possible by Jeffery, so we confirm their conjecture.

        We additionally remark that this extra factor, although being only logarithmic, can indeed be very large. One of the most common usecases for the MNRS-framework is walking over Johnson graphs, which typically have exponentially many vertices. In those cases, $\pi_{\min}$ is inevitably exponentially small, which means that $\log(1/\pi_{\min})$ results in polynomial overhead.

        Moreover, avoiding the extra overhead brings the results in line with the results from \cite{carette2020extended}, who show that one can amortize the setup, update, and checking query complexities if these routines are constructed using learning graphs. This work generalizes that result to the case where these routines can be any feasible solution to the adversary bound.
    \end{enumerate}

    \subsection{Future directions}

    In the variable-time quantum walk search results that we obtain in this work are similar in flavor to those from \cite{jeffery2022quantum}, in the sense that the cost of \textit{traversing} the edges is amortized. However, in some cases, it is the cost of \textit{selecting} the next vertex to traverse to that we would like to traverse. We leave whether such an amortization can be achieved using our techniques for future research.

    Furthermore, one of the most interesting problems in the context of quantum walks is the glued-trees problem, which obtains a provable exponential speed-up over its classical counterpart~\cite{childs2003exponential,belovs2024global}. It would be interesting to figure out if we can recover this exponential speed-up from the generalized graph composition framework, but we leave this for future work as well.

    \subsection{Organization}

    We recall relevant results from the literature in \cref{sec:preliminaries}. We present our improvement analysis of the effective resistance in \cref{sec:effective-resistance}. We present our explicit implementation of transducers from the adversary bound for state conversion in \cref{sec:transducer}. Then, we turn to the generalized graph composition framework in \cref{sec:generalized-graph-composition}. Next, we show how it unifies the decision-tree frameworks in \cref{sec:decision-trees}, and to the divide-and-conquer framework in \cref{sec:divide-and-conquer}. Finally, we apply the generalized graph composition framework to the quantum walk search problem in \cref{sec:quantum-walk-search}.

    \section{Preliminaries}
    \label{sec:preliminaries}

    \subsection{Notation}

    We write the set of natural numbers $\N = \{1,2,\dots\}$. Whenever we have an expression involving the symbol $\pm$, we interpret this as a pair of expressions where we put replace it with $+$ and $-$, respectively. If we use the symbol $\pm$ multiple times in the same expressions, we replace each with the same sign, and if we use $\mp$ instead, we use the opposite sign. For example, $5 \pm 3 \mp 3 = 5$ represents the two statements $5 + 3 - 3 = 5$ and $5 - 3 + 3 = 5$.

    Let $\Omega$ be a finite set. We let $\Delta_{\Omega}$ be the set of probability distributions on $\Omega$. We write $\mathbf{e}_\omega \in \R^\Omega$ for the standard basis vector with its $1$ at the position with index $\omega$.

    Let $G = (V,E)$ be an undirected graph. We associate to every edge $e \in E$ an implicit direction, i.e., a starting and ending vertex $e_- \in V$ and $e_+ \in V$, respectively. For all $v \in V$, we write $N_-(v) \subseteq E$ for all outgoing edges and $N_+(v) \subseteq E$ for all incoming edges. We also write $N(v) = N_-(v) \cup N_+(v)$. As such, $N(v) \cap N(w)$ is the set of all edges connecting $v$ and $w$.

    Let $d \in \N$ and $f,g : \R_{\geq0}^d \supseteq \D \to \R$. We say that $f \in O(g)$ if there exists a $C,M > 0$ such that for all $x \in \D$ with $\norm{x} \geq M$, then $|f(x)| \leq C \cdot |g(x)|$. We write $f \in \Omega(g)$ if and only if $g \in O(f)$, and we write $\Theta(f) = O(f) \cap \Omega(f)$. Furthermore, we write $f \in \widetilde{O}(g)$, if there exists a $k \in \N$ such that $f \in O(g \cdot \log^k(g))$. Similarly, we write $f \in \widetilde{\Omega}(g)$ if and only if $g \in \widetilde{O}(f)$, and $\widetilde{\Theta}(f) = \widetilde{O}(f) \cap \widetilde{\Omega}(f)$. Finally, we write $f \in O(g \cdot \polylog((x_j)_{j \in S}))$, for $S \subseteq [d]$, if there exist $k_j \in \N$, for every $j \in S$, such that $f \in O(g(x) \cdot \prod_{j \in S} \log^{k_j}(x_j))$.

    \subsection{Markov chains}

    We revisit the theory of Markov chains here. We base our exposition on \cite{levin2017markov}, but there exist plenty of alternative resources, e.g., \cite{haggstrom2002finite,klenke2008probability}.

    \begin{definition}[{\cite[Section~1.1]{levin2017markov}}]
        A Markov process $M = (\Omega,P)$ is defined by:
        \begin{enumerate}[nosep]
            \item A finite state space $\Omega$.\footnote{In this work, we will always assume that $\Omega$ is finite. There is also a vast literature on Markov chains that considers infinite state spaces, or even continuous ones, but this is beyond the scope of our work.}
            \item A probability transition matrix $P \in \R_{\geq0}^{\Omega \times \Omega}$, such that every row sums to $1$.
        \end{enumerate}
        If we take a random variable $X_0$ according to some probability distribution on $\Omega$, and then select a sequence of random variables $X_1, X_2, \dots$ such that for all $t \in \N$ and $\omega_1,\omega_2 \in \Omega$,
        \[\P[X_t = \omega_2 | X_{t-1} = \omega_1] = P_{\omega_1,\omega_2},\]
        then we say that $X_0, X_1, X_2, \dots$ is a Markov chain. We say that the Markov chain transition from state $\omega_1 \in \Omega$ into state $\omega_2 \in \Omega$ with probability $P_{\omega_1,\omega_2}$.
    \end{definition}

    We discuss some properties of Markov chains.

    \begin{definition}[{\cite[Sections~1.3 and 1.5]{levin2017markov}}]
        Let $M = (\Omega,P)$ be a Markov process.
        \begin{enumerate}[nosep]
            \item A distribution $\pi \in \Delta_{\Omega}$ is a \textit{stationary distribution} for $M$ if $\pi = \pi P$.
            \item A Markov chain is \textit{irreducible} if it is possible to reach any state from any other state through a sequence of transitions with non-zero probability. In other words, for all $\omega_1, \omega_2 \in \Omega$, there exists a $t \in \N$ such that $(P^t)_{\omega_1,\omega_2} > 0$.
        \end{enumerate}
    \end{definition}

    \begin{theorem}[{\cite[Corollary~1.17]{levin2017markov}}]
        \label{thm:pi-unique}
        If a Markov chain is irreducible, it has a unique stationary distribution, and $\pi_{\omega} > 0$ for all $\omega \in \Omega$.
    \end{theorem}

    A well-studied type of Markov processes are reversible. We introduce them formally here.

    \begin{definition}[{\cite[Section~1.6]{levin2017markov}}]
        Let $M = (\Omega,P)$ be a Markov process with stationary distribution $\pi$. We say that $M$ is reversible with stationary distribution $\pi$ if for all $\omega_1,\omega_2 \in \Omega$, we have
        \[\pi_{\omega_1}P_{\omega_1,\omega_2} = \pi_{\omega_2}P_{\omega_2,\omega_1}.\]
        This condition is referred to as the detailed balance condition. If $M$ is irreducible, then $\pi$ is unique by \cref{thm:pi-unique} and so we simply say that $M$ is reversible.
    \end{definition}

    Note that this condition is a severe restriction, as it reduces the total number of degrees of freedom in $P$ significantly. Nevertheless, this type of Markov processes is most important to us, because it has a direct connection to random walks on graphs, which we discuss next.

    \subsection{Random walks}
    \label{subsec:random-walks}

    We now focus our attention to random walks on graphs.

    \begin{definition}[{\cite[Section~9.1]{levin2017markov}}]
        Let $G = (V,E)$ be an undirected graph, with resistances $r : E \to \R_{>0}$. Now, let $P \in \R_{\geq0}^{\Omega \times \Omega}$ such that
        \[P_{vw} = \sum_{e \in N(v) \cap N(w)} \frac{r_v}{r_e}, \qquad \text{where} \qquad r_v = \left[\sum_{e \in N(v)} \frac{1}{r_e}\right]^{-1}.\]
        The Markov process $M = (V,P)$ is the \textit{random walk on $G$ with resistances $r$}.
    \end{definition}

    We readily verify that $M$ is indeed a valid Markov process, because all the row-sums of $P$ are indeed $1$. We now derive some interesting properties of these random walks.

    \begin{theorem}[{\cite[Section~9.1]{levin2017markov}}]
        Let $G = (V,E)$ be an undirected graph with resistances $r : E \to \R_{>0}$, and let $M = (V,P)$ be the random walk on $G$ with resistances $r$.
        \begin{enumerate}[nosep]
            \item Let $\pi \in \Delta_V$ be defined by
            \[\pi_v = \frac{R}{r_v}, \qquad \text{where} \qquad R = \left[\sum_{v \in V} \frac{1}{r_v}\right]^{-1}.\]
            Then, $M$ is reversible with stationary distribution $\pi$.
            \item $G$ is connected if and only if $M$ is irreducible. In that case, $\pi$ is unique.
        \end{enumerate}
    \end{theorem}

    We see from the above theorem that random walks on connected undirected graphs generate irreducible, reversible Markov processes. It turns out that the converse is also true.

    \begin{theorem}[{\cite[Section~9.1]{levin2017markov}}]
        \label{thm:random-walk-characterization}
        Let $M = (V,P)$ be an irreducible, reversible Markov process, with stationary distribution $\pi$. Let $G = (\Omega,E)$ be defined as
        \[E = \{\{v,w\} : v,w \in V, P_{vw} > 0\}, \qquad \text{and} \qquad r_{\{v,w\}} = \frac{1}{\pi_vP_{vw}}.\]
        Let $M'$ be the random walk on $G$ with resistances $r$. Then, $M = M'$.
    \end{theorem}

    Note that the above is well-defined, because we know from the detailed balance condition that $P_{vw} > 0$ if and only if $P_{wv} > 0$, for any two states $v,w \in V$. Similarly, the definition or $r_{\{v,w\}}$ makes sense because the expression on the right-hand side is symmetric in $v$ and $w$ due to the detailed balance condition.

    \subsection{Electrical networks}

    We can additionally think about undirected graphs with edge-resistances as electrical networks. We follow the exposition in \cite{cornelissen2025quantum}, but very similar expositions can be found in \cite{levin2017markov}.

    \begin{definition}[{\cite[Section~4.1]{cornelissen2025quantum}}]
        \label{def:effective-resistance}
        Let $G = (V,E)$ be an undirected graph, with resistances $r : E \to \R_{>0}$.
        \begin{enumerate}[nosep]
            \item A flow is a function $f : E \to \R$. We embed all flows in the flow space $\C^E = \Span\{\ket{e} : e \in E\}$ by
            \[\ket{f} = \sum_{e \in E} f_e\sqrt{r_e}\ket{e}.\]
            We refer to $\norm{\ket{f}}^2$ as the energy dissipated by $f$, or simply the energy of $f$.
            \item For any flow $f$, we define the net-flow $\delta_f : V \to \R$ as
            \[\delta_f(v) = \sum_{e \in N_-(v)} f_e - \sum_{e \in N_+(v)} f_e.\]
            If $\delta_f \equiv 0$, then $f$ is a circulation, and we define the circulation space as $\mathcal{C} = \{\ket{f} : f \text{ is a circulation}\}$.
            \item For any net-flow $\delta : V \to \R$, the set of corresponding flows $f$ such that $\delta_f = \delta$ forms an affine subspace, i.e., $\ket{f_\delta^{\min}} + \mathcal{C}$. We refer to $f_{\delta}^{\min}$ as the minimum-energy flow with net-flow $\delta$, and we write
            \[R(G,r;\delta) = \norm{\ket{f_{\delta}^{\min}}}^2.\]
            We refer to this as the effective resistance of net-flow $\delta$ in $G$ with resistances $r$.
        \end{enumerate}
    \end{definition}

    We observe that our definition of effective resistance agrees with the usual definition between two nodes $v,w \in V$. Indeed, if we take $\delta(v) = 1$, $\delta(w) = -1$, and we take $\delta \equiv 0$ everywhere else, then $R(G,r;\delta)$ is indeed the effective resistance between $v$ and $w$ in $G$, as witnessed by Thomson's principle (see e.g., \cite[Theorem~9.10]{levin2017markov}). As such, one can view our definition of effective resistance as a generalization of the usual version to more general net-flows $\delta$.

    There exists a very natural way to compute the effective resistance through the graph Laplacian. We loosely follow the exposition as introduced in \cite{mieghem2023graph}.

    \begin{definition}[{\cite[Section~2.1]{mieghem2023graph}}]
        Let $G = (V,E)$ be an undirected graph with resistance $r : E \to \R_{>0}$.
        \begin{enumerate}[nosep]
            \item The weighted incidence matrix is the matrix $B \in \R^{V \times E}$ is defined as
            \[B_{ve} = \begin{cases}
                \frac{1}{\sqrt{r_e}}, & \text{if } e_- \neq e_+ \land v = e_-, \\
                -\frac{1}{\sqrt{r_e}}, & \text{if } e_- \neq e_+ \land v = e_+, \\
                0, & \text{otherwise}.
            \end{cases}\]
            \item The graph Laplacian $L \in \R^{V \times V}$ is defined as $L = BB^T$.
        \end{enumerate}
    \end{definition}

    We now observe that the effective resistance can be computed through the Laplacian.

    \begin{theorem}
        \label{thm:resistance-through-laplacian}
        $R(G,r;\delta) = \delta^TL^+\delta$, where $L^+$ is the Moore-Penrose pseudo-inverse of $L$.
    \end{theorem}

    One can find the proof of this theorem in \cite[Section~5.1]{mieghem2023graph} for the case where $\delta = e_v - e_w$, for $v,w \in V$, but we need this result for arbitrary net-flows $\delta$. This is why we provide a self-contained proof here.

    \begin{proof}[Proof of \cref{thm:resistance-through-laplacian}]
        It is elementary to check that for any flow $f$, we have $B\ket{f} = \delta_f$. As such, the minimum-energy flow $f$ that has net-flow $\delta$ satisfies $\ket{f_{\delta}^{\min}} = B^+\delta$. Now, we obtain
        \[R(G,r;\delta) = \braket{f_{\delta}^{\min}}{f_{\delta}^{\min}} = \delta^T(B^+)^TB^+\delta = \delta^T(BB^T)^+\delta = \delta^TL^+\delta.\qedhere\]
    \end{proof}

    For future reference, we relate the Laplacian of the graph that is generated from an irreducible, reversible Markov process to its probability-transition matrix.

    \begin{theorem}[{\cite[Section~3.11]{mieghem2023graph}}]
        \label{thm:laplacian-and-prob-transition}
        Let $M = (V,P)$ be an irreducible, reversible Markov process with stationary distribution $\pi$. Let $G$ be the corresponding graph with resistances $r$. Let $L$ be its Laplacian. Then,
        \[L = \diag(\pi)(I - P) \qquad \Leftrightarrow \qquad P = I - \diag(\pi)^{-1}L.\]
    \end{theorem}

    \subsection{Quantum adversary bound for state-conversion and transducers}

    We use the formulation of the quantum adversary bound for state conversion from \cite[Section~8.1]{belovs2024taming}. This semi-definite optimization program was first introduced in \cite[Definition~3.2]{lee2011quantum}, also appeared in \cite[Definition~8]{belovs2015variations}, and subsequently in \cite[Definition~7.3]{belovs2023one}.

    \begin{definition}[Quantum adversary bound for state conversion]
        Let $\D$ be some finite domain, and for all $x \in \D$, we define $\ket{\sigma_x} \in \V$, $\ket{\tau_x} \in \V$ and $O_x \in \mathcal{L}(\mathcal{M})$. Then, $P = \{(\ket{\sigma_x},\ket{\tau_x},O_x)\}_{x \in \D}$ is a state-conversion problem, and its adversary bound is the following optimization program:
        \begin{align*}
            \min\quad & \max_{x \in \D} \norm{\ket{w_x}}^2, \\
            \text{s.t.}\quad & \braket{\sigma_x}{\sigma_y} - \braket{\tau_x}{\tau_y} = \bra{w_x}\left(\left(I_{\mathcal{M}} - O_x^{\dagger}O_y\right) \otimes I_{\mathcal{W}}\right)\ket{w_y}, & \forall (x,y) \in \D, \\
            & \ket{w_x} \in \mathcal{M} \otimes \mathcal{W}, \text{ where } \mathcal{W} \text{ is a Hilbert space}, & \forall x \in \D.
        \end{align*}
        We denote the optimal value by $\ADV(P)$. If all the states satisfy $\norm{\ket{\sigma_x}} = \norm{\ket{\tau_x}} = 1$, for all $x \in \D$, then we say that this is a unit-norm state-conversion problem, and if the operations $O_x$ are unitaries, we say that $P$ is a unitary-oracle state-conversion problem. For a function $f : \Sigma_1 \to \Sigma_2$, we define $\ket{\sigma_x} = \ket{\perp}$ and $\ket{\tau_x} = \ket{f(x)}$, for all $x \in \D$, and we abbreviate the optimal value to $\ADV(\{(\ket{\bot}, \ket{f(x)}, O_x)\}_{x \in \D}) = \ADV(f, \{O_x\}_{x \in \D})$.
    \end{definition}

    Note that this definition is subtly different from other versions of the adversary bound that can be found in the existing literature. We refer to \cite[Section~8.2]{belovs2024taming} for a discussion on why these differ by a constant factor of $2$.

    The relevance of the quantum adversary bound is highlighted by the fact that one can convert a feasible solution to a transducer.

    \begin{definition}[{\cite{belovs2024taming}}]
        A \textit{transducer} $U$ is a unitary on $\V \oplus \H$. Let $\ket{\sigma},\ket{\tau} \in \V$ and $\ket{w} \in \H$. If $U$ performs the mapping $\ket{\sigma} \oplus \ket{w} \overset{U}{\mapsto} \ket{\tau} \oplus \ket{w}$, then we say that $U$ transduces $\ket{\sigma}$ to $\ket{\tau}$, and we write $\ket{\sigma} \overset{U}{\rightsquigarrow} \ket{\tau}$, and we refer to $\ket{w}$ as the catalyst for this transduction.
    \end{definition}

    We now borrow the following properties of transducers.

    \begin{theorem}[{\cite[Theorem~5.1]{belovs2024taming}}]
        Let $U$ be a transducer on $\V \oplus \H$. Then, for every $\ket{\sigma} \in \V$, there exists a unique $\ket{\tau} \in \V$ such that $\ket{\sigma} \overset{U}{\rightsquigarrow} \ket{\tau}$. Moreover, the minimal catalyst, i.e., the catalyst $\ket{w} \in \H$ for this transduction of minimal norm is uniquely defined. Finally, the mappings $U_{\transduce\V} : \ket{\sigma} \mapsto \ket{\tau}$ and $U_{\rightsquigarrow\V} : \ket{\sigma} \mapsto \ket{w}$ are linear.
    \end{theorem}

    Finally, we observe that transduction actions can be emulated in quantum algorithms, and that the efficiency of this operation depends on the norm of the catalyst state.

    \begin{theorem}[{\cite[Theorem~5.5]{belovs2024taming}}]
        Let $U$ be a transducer on $\V \oplus \H$, and $K > 0$. Then, there exists a quantum circuit $\A$ that calls $U$ a total of $K$ times, and that implements the mapping
        \[\norm{\A\ket{\sigma} - U_{\transduce\V}\ket{\sigma}} \leq \frac{2\norm{U_{\rightsquigarrow\V}\ket{\sigma}}}{\sqrt{K}}.\]
    \end{theorem}

    Next, it is shown that one can turn an adversary bound solution into a transducer.

    \begin{theorem}[{\cite[Section~8.1]{belovs2024taming}}]
        Let $\{\sigma_x\}_{x \in \D} \subseteq \V_1$, $\{\tau_x\}_{x \in \D} \subseteq \V_2$ and $\{O_x\}_{x \in \D} \subseteq \mathcal{L}(\mathcal{M})$ define a state-conversion problem $P$. Let $\{\ket{w_x}\}_{x \in \D} \subseteq \mathcal{M} \otimes \mathcal{W}$ be a feasible solution to the quantum adversary bound for state conversion. Then, there exists a unitary $U$ acting on $(\V_1 \oplus \V_2) \oplus (\mathcal{M} \otimes \mathcal{W})$, such that $U_x := U \cdot ((I_{\V_1} \oplus I_{\V_2}) \oplus (O_x \otimes I_{\mathcal{W}}))$ acts as
        \[(U_x)_{\transduce(\V_1 \oplus \V_2)} : \ket{\sigma_x} \mapsto \ket{\tau_x}, \qquad \text{and} \qquad (U_x)_{\rightsquigarrow(\V_1 \oplus \V_2)} : \ket{\sigma_x} \mapsto \ket{w_x}.\]
        Consequently, $\mathsf{Q}(P) \in O(\ADV(P))$.
    \end{theorem}

    \section{Effective resistances of Markov processes}
    \label{sec:effective-resistance}

    In this section, we explore the connection between Markov processes and the effective resistance. We start by defining some shorthand notation:

    \begin{definition}
        Let $M = (V,P)$ be an irreducible, reversible Markov process, and let $G = (V,E)$ be the corresponding undirected graph with resistance function $r$, as in \cref{thm:random-walk-characterization}. For any net-flow $\delta \in \R^V$, we write $R_{\mathrm{eff}}(P;\delta) = R_{\mathrm{eff}}(G,r;\delta)$. We also write
        \begin{enumerate}[nosep]
            \item For all $v,w \in V$, we write $R_{\mathrm{eff}}(G,r;v \leftrightarrow w) = R_{\mathrm{eff}}(G,r;\mathbbm{1}_v - \mathbbm{1}_w)$, i.e., the effective resistance between vertices $v$ and $w$.
            \item For any two distributions $\mu,\nu \in \Delta_V$, we write $R_{\mathrm{eff}}(G,r;\mu \leftrightarrow \nu) = R_{\mathrm{eff}}(G,r;\mu - \nu)$.
            \item For a distribution $\mu \in \Delta_V$ and a set $M \subseteq V$, we write
            \[R_{\mathrm{eff}}(G,r;\mu \leftrightarrow M) = \min_{\nu \in \Delta_V : \nu|_{V \setminus M} \equiv 0} R_{\mathrm{eff}}(G,r;\mu \leftrightarrow \nu).\]
        \end{enumerate}
        We extend all these notations $R(G,r;\cdot)$ to $R(P;\cdot)$ as well.
    \end{definition}

    The effective resistance $R_{\mathrm{eff}}(P; v \leftrightarrow w)$ between two distinct vertices $v,w \in V$ can be interpreted combinatorially as the ``commute time'' between $v$ and $w$, i.e., the expected time it takes to start a random walk in $v$, hit $w$, and then come back to $v$~\cite[Proposition~10.7]{levin2017markov}. Similarly, the effective resistance $R_{\mathrm{eff}}(P; \pi \leftrightarrow M)$ can be interpreted as the expected time to hit $M$ starting from a vertex sampled from $\pi$, where $\pi$ is the stationary distribution of the random walk.

    However, Apers, Gily\'en and Jeffery realized that \cite{apers2021unified} these combinatorial interpretations fail in more complicated settings. For instance, they showed that the effective resistance $R_{\mathrm{eff}}(P; \mu \leftrightarrow M)$, with $\mu \in \Delta_S$ and $S,M \subseteq V$, cannot be interpreted as some commute time between $S$ and $M$. This inherently limits our analysis of the effective resistance, especially in cases where the net-flow $\delta$ does not bear any apparent structure.

    How these combinatorial interpretations come into play can be very nicely illustrated when we consider fast-forwarding walks. Indeed, suppose that we have a Markov process $M = (V,P)$, and now we consider $M^t := (V,P^t)$, i.e., the process that takes $t > 1$ steps of the original walk. Using combinatorial arguments, we can now easily prove that
    \[R_{\mathrm{eff}}(P; v \leftrightarrow w) \leq t \cdot R_{\mathrm{eff}}(P^t; v \leftrightarrow w).\]
    Indeed, if we interpret the left- and right-hand side as commute times, then the inequality is obvious, because we are essentially considering the same random process on both sides of the equation, but on the left-hand side we are checking at every step whether we hit the vertex $w$ or returned back to $v$, whereas on the right-hand side we check this only once every $t$ steps. However, we cannot use these combinatorial arguments to prove, for instance, that $R_{\mathrm{eff}}(P; \delta) \leq t \cdot R_{\mathrm{eff}}(P^t; \delta)$, for any net-flow $\delta$.

    To mitigate this, we drop the combinatorial interpretation of the effective resistance, and interpret it linear-algebraically, i.e., using its relation to the graph Laplacian and the probability transition matrix~ \cref{thm:resistance-through-laplacian,thm:laplacian-and-prob-transition}. This allows us to obtain the following result, which we expect to be independently interesting:

    \begin{lemma}
        \label{lem:fast-forwarding}
        Let $M = (V,P)$ be an irreducible, reversible Markov process, with stationary distribution $\pi$. Let $\xi \in \R_V$ such that $\sum_{v \in V} \xi_v = 0$. We write $D = \diag(\pi)$, and $\widetilde{\xi} = D^{-1/2}\xi$. Then, we have the following two consequences:
        \begin{enumerate}[nosep]
            \item For integer $t \geq 1$, $R_{\mathrm{eff}}(P;\xi) \leq t \cdot R_{\mathrm{eff}}(P^t;\xi)$.
            \item If the spectral gap of $P$ is $\delta$, we have $R_{\mathrm{eff}}(P;\xi) \leq \|\widetilde{\xi}\|_2^2/\delta$.
        \end{enumerate}
    \end{lemma}

    \begin{proof}
        Recall from \cref{thm:pi-unique} that the stationary distribution $\pi$ is uniquely defined. Next, recall from \cref{thm:laplacian-and-prob-transition} that $L = D(I-P) = D^{1/2}(I - \widetilde{P})D^{1/2}$ is the Laplacian of the graph that $M$ walks on. We observe that $\widetilde{P} = I - D^{-1/2}LD^{-1/2}$ is symmetric, and so we can diagonalize $\widetilde{P}$ and write it as
        \[\widetilde{P} = \sum_{j=1}^r \lambda_j v_jv_j^T, \qquad \text{which implies} \qquad (I - \widetilde{P})^+ = \sum_{\substack{j=1 \\ \lambda_j < 1}}^r \frac{v_jv_j^T}{1 - \lambda_j}.\]
        Thus, we obtain that
        \[R_{\mathrm{eff}}(P;\xi) = \xi^TL^+\xi = \widetilde{\xi}^T(I - \widetilde{P})^+\widetilde{\xi} = \sum_{\substack{j=1 \\ \lambda_j < 1}}^r \frac{|\widetilde{\xi}^Tv_j|^2}{1 - \lambda_j}.\footnote{This expression is similar to the expression considered in \cite[Proposition~9]{krovi2016quantum}. However, here we consider a more general setting where $\xi$ can be any vector with vanishing sum of all entries, whereas in \cite{krovi2016quantum} only uniform superpositions are considered.}\]

        Now, for the first claim, let $t \geq 1$ be an integer, and observe that $\widetilde{P}^t = D^{-1/2}P^tD^{1/2}$. Using $1/(1-x) \leq t/(1-x^t)$, for all $x < 1$, we observe that
        \[R_{\mathrm{eff}}(P;\xi) = \sum_{\substack{j=1 \\ \lambda_j < 1}}^r \frac{|\widetilde{\xi}^Tv_j|^2}{1 - \lambda_j} \leq t \cdot \sum_{\substack{j=1 \\ \lambda_j < 1}}^r \frac{|\widetilde{\xi}^Tv_j|^2}{1 - \lambda_j^t} = t \cdot R_{\mathrm{eff}}(P^t;\xi).\]

        For the second claim, observe that for all $\lambda_j < 1$, we have $\lambda_j \leq 1-\delta$. Thus, we obtain
        \[R_{\mathrm{eff}}(P;\xi) = \sum_{\substack{j=1 \\ \lambda_j < 1}}^r \frac{|\widetilde{\xi}^Tv_j|^2}{1 - \lambda_j} \leq \frac{1}{\delta} \sum_{\substack{j=1 \\ \lambda_j < 1}}^r |\widetilde{\xi}^Tv_j|^2 \leq \frac{\|\widetilde{\xi}\|_2^2}{\delta}.\qedhere\]
    \end{proof}

    Along similar lines, we prove a lower bound on the effective resistance using linear-algebraic methods.

    \begin{lemma}
        \label{lem:fraction-vs-eff-resistance}
        Let $t \geq 1$ be an integer. Let $M = (V,P)$ be an irreducible, reversible Markov process, with stationary distribution $\pi \in \Delta_V$. Let $\sigma, \nu \in \Delta_V$ have disjoint support. Then, $\sum_{v \in V} \nu_v^2/\pi_v \leq 2R_{\mathrm{eff}}(P^t; \sigma - \nu)$.
    \end{lemma}

    \begin{proof}
        Recall from \cref{thm:laplacian-and-prob-transition} that $D(I-P) = L \succeq 0$. Thus, we have $D^{1/2}(I - D^{1/2}PD^{-1/2})D^{1/2} \succeq 0$, and by consequence also $I - D^{1/2}PD^{-1/2} \succeq 0$. Moreover, the eigenvalues of $D^{1/2}PD^{-1/2}$ are the same as those of $P$, because they are related by a similarity transform. We also have $\norm{P} \leq 1$, and so we find that all these eigenvalues are in the interval $[-1,1]$. Combining everything implies that the eigenvalues of $I - D^{1/2}PD^{-1/2}$ are contained in the interval $[0,2]$. Thus, all the eigenvalues of its Moore-Penrose pseudo-inverse are at least $1/2$, and since the unique $0$-eigenvector is $\sqrt{\pi}$, we find that $(I - D^{-1/2}PD^{-1/2})^+ \succeq (I - \sqrt{\pi}\sqrt{\pi}^T)/2$. Thus, we find that
        \begin{align*}
            R_{\text{eff}}(P^t, \sigma - \nu) &= (\sigma - \nu)^TL^+(\sigma - \nu) = (\sigma - \nu)^T D^{-1/2} (I - D^{1/2}P^tD^{-1/2})^+ D^{-1/2} (\sigma - \nu) \\
            &\geq \frac12(\sigma - \nu)^T D^{-1/2} (I - \sqrt{\pi}\sqrt{\pi}^T) D^{-1/2}(\sigma - \nu) \\
            &= \frac12(\sigma - \nu)^T D^{-1} (\sigma - \nu) - \underbrace{\frac12(\sigma - \nu)\mathbbm{1} \mathbbm{1}^T(\sigma - \nu)}_{= 0} = \frac12\sum_{v \in V} \frac{(\sigma_v - \nu_v)^2}{\pi_v} \geq \frac12\sum_{v \in \supp(\nu)} \frac{\nu_v^2}{\pi_v}.\qedhere
        \end{align*}
    \end{proof}

    \section{Explicit transducers from adversary-bound solutions}
    \label{sec:transducer}

    In \cite[Section~8.1]{belovs2024taming}, Belovs, Jeffery and Yolcu present a way to convert feasible solutions to the dual adversary bound for state conversion into transducers. However, this result only asserts the existence of such a transducer, but the proof is not constructive. In this section, we provide a constructive proof for the same result.

    We start with a structural observation, namely that we can equivalently look at a reformulated version of the state-conversion problem, in which we are reflecting through exactly half of the input states. As such, we conveniently refer to this specific form of a state-conversion problem as a \textit{state-reflection problem}.

    \begin{definition}[State-reflection problem]
        \label{def:state-reflection-problem}
        Let $\D$ be a finite set. For all $x \in \D$, let $\ket{\sigma_x^+}, \ket{\sigma_x^-}$ be two orthogonal vectors in $\V$, and let $O_x$ be an operation on $\H$, such that $O_x^2 = I$. We refer to the state-conversion problem $\{(\ket{\sigma_x^{\pm}}, \pm\ket{\sigma_x^{\pm}}, O_x)\}_{x \in \D}$ as a state-reflection problem, and we write it as $R = \{(\ket{\sigma_x^+}, \ket{\sigma_x^-}, O_x)\}_{x \in \D}$ for short. Let $w := \{\ket{w_x^{\pm}}\}_{x \in \D}$ be a feasible solution to $R$, satisfying $O_x\ket{w_x^{\pm}} = \pm\ket{w_x^{\pm}}$. We refer to $\ket{w_x^+}$ and $\ket{w_x^-}$ as the positive and negative witnesses for $x$, and similarly to $\norm{\ket{w_x^+}}^2$ and $\norm{\ket{w_x^-}}^2$ as the positive and negative witness sizes for $x \in \D$, which write as $\mathsf{R}^+_x(w)$ and $\mathsf{R}^-_x(w)$.
    \end{definition}

    We now show that every unit-norm unitary-oracle state-conversion problem can be equivalently reformulated as a state-reflection problem. This equivalence only holds if we have access to both forward and backwards versions of the oracle $O_x$ and $O_x^{\dagger}$, though, in other cases the situation might be more subtle.

    \begin{lemma}
        \label{thm:reformulation-state-conversion}
        Let $P = \{(\ket{\sigma_x}, \ket{\tau_x}, O_x)\}_{x \in \D}$ be a unit-norm unitary-oracle state-conversion problem, with input and output and oracle spaces $\V_1$, $\V_2$ and $\H$, respectively. For all $x \in \D$, we write
        \[\ket{\sigma_x^{\pm}} = \frac{\ket{\sigma_x} \oplus \pm\ket{\tau_x}}{\sqrt{2}} \in \V_1 \oplus \V_2, \qquad \text{and} \qquad \overline{O_x} = (X \otimes I)(O_x \oplus O_x^{\dagger}) \in \mathcal{L}(\H \oplus \H),\]
        where $X \otimes I$ acting on $\H \oplus \H$ swaps both copies of $\H$. Then, $R = \{(\ket{\sigma_x^+}, \ket{\sigma_x^-}, \overline{O_x})\}_{x \in \D}$ is a state-reflection problem, $\ADV(P) = \ADV(R)$, and there exists an optimal solution $\{\ket{w_x^{\pm}}\}_{x \in \D}$ to $R$ such that $\overline{O_x}\ket{w_x^{\pm}} = \pm\ket{w_x^{\pm}}$.
    \end{lemma}

    \begin{proof}
        We first show that $R$ satisfies the conditions of a state-reflection problem, as outlined in \cref{def:state-reflection-problem}. To that end, we check that $\ket{\sigma_x^+}$ and $\ket{\sigma_x^-}$ are orthogonal for all $x \in D$, since
        \[\braket{\sigma_x^+}{\sigma_x^-} = \left(\frac{\bra{\sigma_x} \oplus \bra{\tau_x}}{\sqrt{2}}\right)\left(\frac{\ket{\sigma_x} \oplus \ket{\tau_x}}{\sqrt{2}}\right) = \frac12\left(\braket{\sigma_x}{\sigma_x} - \braket{\tau_x}{\tau_x}\right) = 0.\]
        Next, we observe that $\overline{O_x}^2 = I$, for all $x \in \D$, as
        \[\overline{O_x}^2 = (X \otimes I)(O_x \oplus O_x^{\dagger})(X \otimes I)(O_x \oplus O_x^{\dagger}) = (X \otimes I)(X \otimes I) \cdot (O_x^{\dagger} \oplus O_x)(O_x \oplus O_x^{\dagger}) = I.\]
        Thus, $R$ indeed is a state-reflection problem.

        Next, we prove that the adversary values are the same for $P$ and $R$. To that end, let $\{\ket{w_x}\}_{x \in \D}$ be a feasible solution to $P$. Now, for all $x \in \D$, we define
        \[\ket{w_x^{\pm}} := \frac{\ket{w_x} \oplus \pm O_x\ket{w_x}}{\sqrt{2}},\]
        and we verify that this is a feasible solution to $R$. First, let $x \in \D$, and observe that
        \[\overline{O_x}\ket{w_x^{\pm}} = \frac{\mathrm{SWAP}(O_x\ket{w_x} \oplus \pm O_x^{\dagger}O_x\ket{w_x})}{\sqrt{2}} = \frac{\pm\ket{w_x} \oplus O_x\ket{w_x}}{\sqrt{2}} = \pm\ket{w_x^{\pm}}.\]
        Next, let $x,y \in \D$, and observe that
        \[\braket{\sigma_x^{\pm}}{\sigma_y^{\pm}} - \left(\pm\bra{\sigma_x^{\pm}}\right)(\pm\ket{\sigma_y^{\pm}}) = \braket{\sigma_x^{\pm}}{\sigma_y^{\pm}} - \braket{\sigma_x^{\pm}}{\sigma_y^{\pm}} = 0,\]
        and on the other hand
        \[\bra{w_x^{\pm}}(I - \overline{O_x}^{\dagger}\overline{O}_y)\ket{w_y^{\pm}} = \braket{w_x^{\pm}}{w_y^{\pm}} - (\pm\bra{w_x^{\pm}})(\pm\ket{w_y^{\pm}}) = \braket{w_x^{\pm}}{w_y^{\pm}} - \braket{w_x^{\pm}}{w_y^{\pm}} = 0.\]
        Thus, it remains to verify what happens for the inner products between $+$ and $-$. To that end, let $x,y \in \D$, and observe that
        \begin{align*}
            &\braket{\sigma_x^+}{\sigma_y^-} - \bra{\sigma_x^+}(-\ket{\sigma_y^-}) = 2\braket{\sigma_x^+}{\sigma_y^-} = 2\left(\frac{\bra{\sigma_x} \oplus \bra{\tau_x}}{\sqrt{2}}\right)\left(\frac{\ket{\sigma_y} \oplus -\ket{\tau_y}}{\sqrt{2}}\right) \\
            &\qquad = \braket{\sigma_x}{\sigma_y} - \braket{\tau_x}{\tau_y} = \bra{w_x}(I - O_x^{\dagger}O_y)\ket{w_y} = 2\braket{w_x^+}{w_y^-} = \bra{w_x^+}(I - \overline{O_x}^{\dagger}\overline{O_y})\ket{w_y^-}.
        \end{align*}
        Thus, $\{\ket{w_x^{\pm}}\}_{x \in \D}$ is indeed a feasible solution to $R$, with the same objective function as $P$, and satisfying $\overline{O_x}\ket{w_x^{\pm}} = \pm\ket{w_x^{\pm}}$. Thus, $\ADV(R) \leq \ADV(P)$.

        It remains to verify the reverse inequality. To that end, let $\{\ket{w_x^{\pm}}\}_{x \in \D}$ be a feasible solution to $R$. Now, for all $x \in \D$ and $b \in \{\pm\}$, we write $\ket{(w_x^{\pm})_b} = (I + b\overline{O_x})\ket{w_x^{\pm}}/2$, i.e., $\ket{(w_x^+)_+}$ is the projection of $\ket{w_x^+}$ onto the $+1$-eigenspace of $\overline{O_x}$, and similarly for the other choices for $b$ and $\pm$. In particular, for all $b \in \{\pm\}$, $\overline{O_x}\ket{(w_x^{\pm})_b} = b\ket{(w_x^{\pm})_b}$, and for all $s,t \in \{\pm\}$, we have $\braket{(w_x^s)_b}{(w_x^t)_{1-b}} = 0$. We let
        \[\ket{w_x} := \frac{(I \oplus O_x^{\dagger} \oplus I \oplus O_x^{\dagger})(\ket{(w_x^+)_0} \oplus \ket{(w_x^+)_1} + \ket{(w_x^-)_1} \oplus \ket{(w_x^-)_0})}{\sqrt{2}}.\]
        Then, for all $x,y \in \D$, we have
        \begin{align*}
            &\bra{w_x}(I - O_x^{\dagger}O_y \otimes I_4)\ket{w_y} = \bra{w_x}(I - O_x^{\dagger}O_y \oplus O_x^{\dagger}O_y \oplus O_x^{\dagger}O_y \oplus O_x^{\dagger}O_y)\ket{w_y} \\
            &\qquad = \frac12\sum_{b,s,t \in \{\pm\}} \bra{(w_x^s)_b}(I \oplus O_x)(I - O_x^{\dagger}O_y \oplus O_x^{\dagger}O_y)(I \oplus O_y^{\dagger})\ket{(w_y^t)_{1-b}} \\
            &\qquad = \frac12\sum_{b,s,t \in \{\pm\}} \bra{(w_x^s)_b}(I \oplus O_x^{\dagger}O_y)\ket{(w_y^t)_{1-b}} - \bra{(w_x^s)_b}(O_xO_y^{\dagger} \oplus I)\ket{(w_y^t)_{1-b}} \\
            &\qquad = \frac12\sum_{b,s,t \in \{\pm\}} \braket{(w_x^s)_b}{(w_y^t)_{1-b}} - \bra{(w_x^s)_b}(O_xO_y^{\dagger} \oplus O_xO_y^{\dagger})\ket{(w_y^t)_{1-b}} \\
            &\qquad = \frac12\sum_{b,s,t \in \{\pm\}} \bra{(w_x^s)_b}(I - \overline{O_x}^{\dagger}\overline{O_y})\ket{(w_y^t)_{1-b}} = \frac12\sum_{b,c,s,t \in \{\pm\}} \bra{(w_x^s)_b}(I - \overline{O_x}^{\dagger}\overline{O_y})\ket{(w_y^t)_c} \\
            &\qquad = \frac12\sum_{s,t \in \{\pm\}} \bra{w_x^s}(I - \overline{O_x}^{\dagger}\overline{O_y})\ket{w_y^t} = \frac12\sum_{s,t \in \{\pm\}} \braket{\sigma_x^s}{\sigma_y^t} - st\braket{\sigma_x^s}{\sigma_y^t} = \braket{\sigma_x^+}{\sigma_x^-} + \braket{\sigma_x^-}{\sigma_x^+} = \braket{\sigma_x}{\tau_x},
        \end{align*}
        and finally we observe that for all $x \in \D$,
        \begin{align*}
            \norm{\ket{w_x}}^2 &= \frac12\left[\norm{\ket{(w_x^+)_0}}^2 + \norm{\ket{(w_x^+)_1}}^2 + \norm{\ket{(w_x^-)_1}}^2 + \norm{\ket{(w_x^-)_0}}^2\right] = \frac12\left[\norm{\ket{w_x^+}}^2 + \norm{\ket{w_x^-}}^2\right] \\
            &\leq \ADV(R),
        \end{align*}
        and so $\ADV(P) \leq \ADV(R)$.
    \end{proof}

    The previous theorem tells us that we can without loss of generality assume that a feasible solution to a state-reflection problem with oracle $O_x$ has the property that for all $x \in \D$, $O_x\ket{w_x^{\pm}} = \pm\ket{w_x^{\pm}}$. We will assume this without further mention in the remainder of this text, and we remark explicitly here that it follows that $\braket{w_x^+}{w_x^-} = 0$, for all $x \in \D$, since they live in different eigenspaces of $O_x$.

    Next, we observe that for a state-reflection problem, the unitary $U$ that is required to implement the transduction action becomes a reflection through a particular subspace. This is the objective of the following theorem.

    \begin{theorem}[Explicit implementation of transducers]
        \label{thm:transducers}
        Let $R = \{(\ket{\sigma_x^{\pm}}, \pm\ket{\sigma_x^{\pm}}, O_x)\}_{x \in \D}$ be a state-reflection problem, with feasible solution $\{\ket{w_x^{\pm}}\}_{x \in \D}$. Let, let $\A \subseteq \V \oplus \H$ be a subspace such that
        \[\Span\left\{\begin{bmatrix}
            \ket{\sigma_x^+} \\
            \ket{w_x^+}
        \end{bmatrix}\right\}_{x \in \D} \subseteq \A \subseteq \Span\left\{\begin{bmatrix}
        \ket{\sigma_x^-} \\
        \ket{w_x^-}
        \end{bmatrix}\right\}_{x \in \D}^{\perp}.\]
        Then, the unitary $U = 2\Pi_A - I$ is such that $U(I \oplus O_x)$ implements $\ket{\sigma_x^{\pm}} \overset{U(I \oplus O_x)}{\rightsquigarrow} \pm\ket{\sigma_x^{\pm}}$.
    \end{theorem}

    \begin{proof}
        For any $x,y \in \D$, we have
        \begin{align*}
            (\bra{\sigma_x^+} \oplus \bra{w_x^+})(\ket{\sigma_y^-} \oplus \ket{w_y^-}) &= \braket{\sigma_x^+}{\sigma_y^-} + \braket{w_x^+}{w_y^-} \\
            &= \frac12\left(\braket{\sigma_x^+}{\sigma_y^-} - \bra{\sigma_x^+}(-\ket{\sigma_y^-}) + \bra{w_x^+}(I - O_x^{\dagger}O_y)\ket{w_y^-}\right) = 0,
        \end{align*}
        and so there indeed exists a subspace $\A$ that satisfies this definition. Moreover, we have
        \[\begin{bmatrix}
            \ket{\sigma_x^{\pm}} \\
            \ket{w_x^{\pm}}
        \end{bmatrix} \overset{I \oplus O_x}{\mapsto} \begin{bmatrix}
            \ket{\sigma_x^{\pm}} \\
            \pm\ket{w_x^{\pm}}
        \end{bmatrix} \overset{U}{\mapsto} \begin{bmatrix}
            \pm\ket{\sigma_x^{\pm}} \\
            \ket{w_x^{\pm}}
        \end{bmatrix}.\qedhere\]
    \end{proof}

    This explicit construction for the unitary $U$ generalizes the explicit implementation of a transducer from a span program~\cite[Theorem~3.3]{cornelissen2025quantum}. We also remark specifically that in the regular graph composition framework, one can take the space $\A$ to be the unit-$st$ flow space, to match the algorithmic implementation in \cite{cornelissen2025quantum}.

    Next, we observe that state-reflection problems exhibit a remarkable amount of flexibility. This is shown in the following lemma.

    \begin{lemma}
        \label{lem:rescale-state-reflection-problem}
        Let $R = \{(\ket{\sigma_x^{\pm}}, \pm\ket{\sigma_x^{\pm}}, O_x)\}_{x \in \D}$ be state-reflection problem. Let $D \in \mathcal{L}(\mathcal{V})$ be an invertible operator, $\alpha_{\pm} \in \C$, and let $R' := \{(\alpha_+D\ket{\sigma_x^+}, \alpha_-(D^{-1})^{\dagger}\ket{\sigma_x^{\pm}}, O_x)\}_{x \in \D}$. Then, any feasible solution $w := \{\ket{w_x^{\pm}}\}_{x \in \D}$ for a state-reflection problem $R$ can be turned into a feasible solution $w' := \{\alpha_+\ket{w_x^+},\alpha_-\ket{w_x^-}\}_{x \in \D}$ for $R'$, where $x^{-\dagger} := (x^{-1})^{\dagger}$. Consequently,
        \[\ADV(R) \leq \sqrt{\max_{x \in \D} \mathsf{R}_x^+(w) \cdot \max_{x \in \D} \mathsf{R}_x^-(w)}.\]
    \end{lemma}

    \begin{proof}
        For all $x,y \in \D$, and $s,t \in \{\pm\}$, we have
        \[\left(\alpha_s\ket{w_x^s}\right)^{\dagger} (I - O_x^{\dagger}O_y) \left(\alpha_t\ket{w_y^t}\right) = \alpha_s^{\dagger}\alpha_t \left(\braket{\sigma_x^s}{\sigma_y^t} - st\braket{\sigma_x^s}{\sigma_y^t}\right) = 2\delta_{s \neq t}(\bra{\sigma_x^s}\alpha_s^{\dagger})(\alpha_t\ket{\sigma_y^t}).\]
        By symmetry, it suffices to consider $s = +$ and $t = -$, and we verify that the right-hand side becomes
        \[2(\bra{\sigma_x^+}\alpha_+^{\dagger})D^{\dagger}(D^{-1})^{\dagger}(\alpha_-\ket{\sigma_y^-}) = 2\left(\bra{\sigma_x^+}D^{\dagger}\alpha_+^{\dagger}\right)\left(\alpha_-(D^{-1})^{\dagger}\ket{\sigma_y^-}\right).\qedhere\]
    \end{proof}

    We remark here that it seems to be quite tricky to see what these rescalings do to the subspace $\A$ from \cref{thm:transducers}. Understanding this seems to be fundamental for time-efficient implementations of the algorithms we consider in this work. We leave this for future work.

    We end this section by observing that we can interpret the construction from Jeffery~\cite{jeffery2022quantum} as a feasible solution to the state-reflection program.

    \begin{theorem}[{\cite[Corollary~3.7]{jeffery2022quantum}}]
        \label{thm:las-vegas-algorithm}
        Let $\A$ be a quantum Monte-Carlo algorithm that implements the state-conversion problem given by $P = \{(\ket{\sigma_x}, \ket{\tau_x}, O_x)\}_{x \in \D}$, with run-time given by the random variable $T_x$, for every $x \in \D$. Then, for any choice of $(\alpha_t)_{t \in \N_0} \subseteq \R_{>0}$, we can implement a feasible solution $W := \{\ket{w_x^{\pm}}\}_{x \in \D}$ to the corresponding state-reflection problem $R := \{(\ket{\sigma_x} \pm \ket{\tau_x}, \pm(\ket{\sigma_x} \pm \ket{\tau_x}), \overline{O_x})\}_{x \in \D}$, with complexities
        \[\mathsf{W}_x^+ = \E\left[\sum_{t=0}^{T_x-1} \alpha_t\right], \qquad \text{and} \qquad \mathsf{W}_x^- = \E\left[\sum_{t=0}^{T_x-1} \frac{1}{\alpha_t}\right].\]
        In particular, by choosing $\alpha_t = 1$, for all $t \in \N_0$, we obtain $\mathsf{W}_x^{\pm} = \E[T_x]$, and by choosing $\alpha_t = t+1$, we obtain $\mathsf{W}_x^+ = \E[T_x(T_x+1)/2] \in O(\E[T_x^2])$, and $\mathsf{W}_x^- \in O(1 + \E[\log(T_x)])$.
    \end{theorem}

    \begin{proof}
        Apply the principle of deferred measurement to $\A$. The resulting algorithm is now a sequence of $U_0, O_x, U_1, \dots, U_T$, with $T \in \N$. For all $t \in [T]_0$, let $\ket{\psi_{x,t}}$ be the state of the algorithm after $U_t$. At every time-step $t \in \N_0$, we let $\H_t \subseteq \H$ be the subspace of the state space on which the control qubits for the application of $O_x$ equal $1$. We define
        \[\ket{w_x^{\pm}} := \bigoplus_{t=0}^{T-1} \sqrt{\alpha_t^{\pm1}} \frac{\Pi_{\H_t}\ket{\psi_{x,t}} \oplus \pm O_x\Pi_{\H_t}\ket{\psi_{x,t}}}{\sqrt{2}}.\]
        Then, we obtain
        \[\mathsf{W}_x^{\pm} = \norm{\ket{w_x^{\pm}}}^2 = \sum_{t=0}^{T-1} \alpha_t^{\pm1}\norm{\Pi_{\H_t}\ket{\psi_{x,t}}}^2 \leq \sum_{t=0}^{T-1} \alpha_t^{\pm1} \P\left[T_x > t\right] = \E\left[\sum_{t=0}^{T_x-1} \alpha_t^{\pm1}\right],\]
        and for all $x,y \in \D$,
        \begin{align*}
            2\braket{w_x^+}{w_y^-} &= \sum_{t=0}^{T-1} \left[\bra{\psi_{x,t}}\Pi_{\H_t}\ket{\psi_{y,t}} - \bra{\psi_{x,t}}O_x^{\dagger}\Pi_{\H_t}O_y\ket{\psi_{y,t}}\right] = \sum_{t=0}^{T-1} \left[\braket{\psi_{x,t}}{\psi_{y,t}} - \bra{\psi_{x,t}}O_x^{\dagger}O_y\ket{\psi_{y,t}}\right] \\
            &= \sum_{t=0}^{T-1} \left[\braket{\psi_{x,t}}{\psi_{y,t}} - \braket{\psi_{x,t+1}}{\psi_{y,t+1}}\right] = \braket{\psi_{x,0}}{\psi_{y,0}} - \braket{\psi_{x,T}}{\psi_{y,T}} = \braket{\sigma_x}{\sigma_y} - \braket{\tau_x}{\tau_y}.\qedhere
        \end{align*}
    \end{proof}

    \section{Generalized graph composition}
    \label{sec:generalized-graph-composition}

    In this section, we generalize the graph composition framework from \cite{cornelissen2025quantum} to the state-conversion setting. We perform the composition on the level of state-reflection problems -- these seem to be more amenable to being composed than their state-conversion problem counterparts.

    \subsection{Hyperedges}

    The core idea is to interpret certain state-reflection problems as hyperedges in a graph, i.e., edges connecting more than two vertices. We start by formally introducing which state-reflection problems can be interpreted in this way.

    \begin{definition}[Hyperedge problem]
        \label{def:hyperedge-problem}
        Let $V$ be a finite set, and $R = \{(\ket{\sigma_x^{\pm}}, \pm\ket{\sigma_x^{\pm}}, O_x)\}_{x \in \D}$ be a state-reflection problem over $\mathcal{V} = \C^V$. Suppose in addition that we can write for all $x \in \D$,
        \[\ket{\sigma_x^+} := \delta_x \in \C^V, \qquad \text{and} \qquad \ket{\sigma_x^-} = U_x \in \C^V,\]
        such that
        \[U_x \text{ is constant on } \supp(\delta_x), \qquad \text{and} \qquad \sum_{v \in V} (\delta_x)_v = 0.\]
        Then, we say that $R$ is a hyperedge problem denoted by $\{(\delta_x,U_x,O_x)\}_{x \in \D}$, and we say that it implements the net-flow $\delta_x$ and the potential function $U_x$ through the hyperedge with vertices $V$.
    \end{definition}

    It is important to note that not every state-reflection problem implements a hyperedge. Indeed, we will see that we need the additional constraints outlined in the above definition to carry out the compositions in the upcoming sections. Nevertheless, we will show that the above definition is still sufficiently generic to implement quantum walks. It would be interesting to figure out if there is a generic conversion from state-reflection problems into this form, but we leave that for future research.

    One can interpret a state-reflection program $R$ implementing a hyperedge on the vertices $V$ pictorially as in \cref{fig:hyperedge}, by drawing parallels to the electrical network framework. The idea is that for every $x \in \D$ there is a net-flow on the vertices in the support of $\delta_x$, and since we require the sum of the entries of $\delta_x$ to vanish, the total flow is conserved. The vertices that have a non-zero net-flow $(\delta_x)_v$, the potential function has to be the same, but there can be some potential function difference with the other vertices adjacent to the hyperedge.

    We observe that one can freely add or subtract constants to the potential functions, analogously to the electric network framework. We formalize that in the following lemma statement.

    \begin{lemma}
        Let $R = \{(\delta_x, U_x, O_x)\}_{x \in \D}$, and for all $x \in \D$, let $C_x \in \C$. Let $R' = \{(\delta_x, U_x + C_x\mathbbm{1}_V, O_x)\}_{x \in \D}$. Then, the feasible region of $R$ equals the feasible region of $R'$.
    \end{lemma}

    \begin{proof}
        We compare the difference of the Gram matrices of the input and output states to the state-reflection problem. For all $x,y \in \D$, we have
        \[\braket{\sigma_x^{\pm}}{\sigma_y^{\pm}} - (\pm\bra{\sigma_x^{\pm}})(\pm\ket{\sigma_y^{\pm}}) = 0,\]
        for any state-reflection problem, so it remains to focus on the cross-talk between $+$ and $-$ states. To that end, let $\ket{(\sigma_x^{\pm})_C}$ be the states in the state-reflection problem for $R'$, and similarly $\ket{\sigma_x^{\pm}}$ for $R$. We observe that
        \[\braket{(\sigma_x^+)_C}{(\sigma_y^-)_C} = \delta_x^T(U_y + C_y\mathbbm{1}_V) = \delta_x^TU_x + C_y\delta_x^T\mathbbm{1}_V = \delta_x^TU_x = \braket{\sigma_x^+}{\sigma_y^-},\]
        and so the feasible regions of $R'$ and $R$ are identical.
    \end{proof}

    Next, we observe that span programs can be converted into state-reflection problems, that implement hyperedges with just two endpoints, which we will refer to as $s$ and $t$. The span program implements a unit flow from $s$ and $t$ for positive instances, and a unit potential difference for negative instances. Importantly, this characterization is an equivalence, so interpreting span programs this way is not lossy, not even up to constant factors.

    \begin{theorem}[Hyperedge from a span program]
        \label{thm:hyperedge-from-span-program}
        Let $f : \D \to \{0,1\}$. For all $x \in \D$, we define net-flow $\delta_x$ and potential function $U_x$ on $\{s,t\}$, as
        \[\delta_x = \begin{cases}
            \mathbbm{1}_s - \mathbbm{1}_t, & \text{if } f(x) = 1, \\
            0, & f(x) = 0,
        \end{cases} \qquad \text{and} \qquad U_x = \begin{cases}
            0, & \text{if } f(x) = 1, \\
            -\mathbbm{1}_t, & f(x) = 0.
        \end{cases}\]

        On the one hand, let $\mathcal{P} = (\H, x \mapsto \H(x), \K, \ket{w_0})$ be a span program on $\D$ that computes $f$. Let $\{\ket{w_x}\}_{x \in \D}$ be a set of witnesses for $\mathcal{P}$. We write
        \[\ket{w_x^+} = \begin{cases}
            \ket{w_x}, & \text{if } f(x) = 1, \\
            0, & \text{if } f(x) = 0,
        \end{cases} \qquad \text{and} \qquad \ket{w_x^-} = \begin{cases}
            0, & \text{if } f(x) = 1, \\
            \ket{w_x}, & \text{if } f(x) = 0.
        \end{cases}\]
        Then, $\{\ket{w_x^{\pm}}\}_{x \in \D}$ is a feasible solution for the hyperedge problem $R = \{(\delta_x, U_x, 2\Pi_{\H(x)} - I)\}_{x \in \D}$.

        Conversely, let $\{\ket{w_x^{\pm}}\}_{x \in \D}$ be a feasible solution to the hyperedge problem $R = \{(\delta_x, U_x, O_x)\}_{x \in \D}$, where the oracles $O_x$ act on $\H$. We let $\H(x)$ be the $+1$-eigenspace of $O_x$. Then, there exists a span program $\mathcal{P} = (\H, x \mapsto \H(x), \K, \ket{w_0})$ that computes $f$, with
        \[C(\mathcal{P})^2 \leq \max\left\{\max_{x \in f^{-1}(1)} \norm{\ket{w_x^+}}^2, \max_{x \in f^{-1}(0)} \norm{\ket{w_x^-}}^2\right\}.\]
    \end{theorem}

    \begin{proof}
        For the first claim, we write $O_x = 2\Pi_{\H(x)} - I$. Since $\ket{w_x} \in \H(x)$ if $f(x) = 1$, and $\ket{w_x} \in \H(x)^{\perp}$ otherwise, we immediately observe that $O_x\ket{w_x^{\pm}} = \pm\ket{w_x^{\pm}}$. As such, for all $x,y \in \D$, we have
        \[\bra{w_x^{\pm}}(I - O_x^{\dagger}O_y)\ket{w_y^{\pm}} = \braket{w_x^{\pm}}{w_y^{\pm}} - \braket{w_x^{\pm}}{w_y^{\pm}} = 0,\]
        which equals $\braket{\sigma_x^{\pm}}{\sigma_y^{\pm}} - (\pm\bra{\sigma_x^{\pm}})(\pm\ket{\sigma_y^{\pm}}) = 0$.

        Thus, it remains to consider the cross-talk between $+$ and $-$. To that end, observe for all $x,y \in \D$ that
        \[\braket{\sigma_x^+}{\sigma_y^-} - \bra{\sigma_x^+}(-\ket{\sigma_y^-}) = 2\braket{\sigma_x^+}{\sigma_y^-} = 2U_x^T\delta_y = 2\delta_{f(x) = 0}\delta_{f(y) = 1},\]
        and on the other hand,
        \[\bra{w_x^+}(I - O_xO_y)\ket{w_y^-} = 2\braket{w_x^+}{w_y^-} = 2\braket{w_x}{w_y}\delta_{f(x) = 1}\delta_{f(y) = 0},\]
        It remains to prove that $\braket{w_x}{w_y} = 1$ whenever $f(x) = 1$ and $f(y) = 0$. Indeed, in that case we can write $\ket{w_x} - \ket{w_0} \in \K$, and $\ket{w_y} \in \K^{\perp}$, and so
        \[\braket{w_x}{w_y} = \braket{w_0}{w_y} + (\bra{w_x} - \bra{w_0})\ket{w_y} = 1 - 0 = 1.\]
        This completes the proof for the first claim.

        For the converse, let
        \[\A = \Span\left\{\begin{bmatrix}
            \ket{\sigma_x^+} \\ \ket{w_x^+}
        \end{bmatrix}\right\}_{x \in \D} = \Span\left\{\begin{bmatrix}
            \ket{\sigma_x^+} \\ \ket{w_x^+}
        \end{bmatrix}\right\}_{x \in f^{-1}(1)}.\]
        Now, since the first two entries of any vector in $\A$ sum to $0$, we can decompose $\A$ into two parts, i.e.,
        \[\A = \Span\left\{\begin{bmatrix}
            \mathbbm{1}_s - \mathbbm{1}_t \\
            \ket{w_0}
        \end{bmatrix}\right\} \oplus \mathcal{K},\]
        where $\ket{w_0} \in \K^{\perp}$. We observe for every $x \in f^{-1}(1)$ that $\ket{w_x^+} - \ket{w_0} \in \K$, and so $\ket{w_0} \in \K + \H(x)$. Similarly, for $x \in f^{-1}(0)$, we have $1 - \braket{w_x^-}{w_0} = 0$, and so the span program computes $f$, and the witnesses are the same as for the state-reflection problem.
    \end{proof}

    There are also other ways to generate hyperedges. We specifically remark here that function evaluations and database updates fit neatly in this context, as witnessed by the following theorem.

    \begin{theorem}[Hyperedge from database update and function evaluation]
        \label{thm:hyperedge-from-database-function}
        Let $\D$ be a finite domain, and consider a database update problem with input $\xi_x \in \Sigma_I$ output $\eta_x \in \Sigma_O$ and oracle $O_x$ satisfying $O_x^2 = I$, for all $x \in \D$. Then, the state-reflection problem $R = \{(\ket{\xi_x} - \ket{\eta_x}, \ket{\xi_x} + \ket{\eta_x}, O_x)\}_{x \in \D}$ implementing the database update problem is also a hyperedge problem.
    \end{theorem}

    \begin{proof}
        We verify that the properties of \cref{def:hyperedge-problem} are satisfied. To that end, we write $\delta_x := \mathbbm{1}_{\xi_x} - \mathbbm{1}_{\eta_x}$ and $U_x := \mathbbm{1}_{\xi_x} + \mathbbm{1}_{\eta_x}$, and we observe that the sum of $\delta_x$ vanishes, and that $U_x$ is constant on $\supp(\delta_x)$.
    \end{proof}

    It seems in general harder to turn a general state-conversion problem into a hyperedge. However, there are some other special cases we can also handle, most notably when we know that the flow vectors are some restriction of a distribution to an arbitrary subset.

    \begin{theorem}
        \label{thm:known-marked-fraction-recovery}
        Let $V$ be a finite set, and let $\pi \in \Delta_V$. For all $x \in \D$ in a finite domain $\D$, let $M_x \subseteq V$, such that $\P_{v \sim \pi}[v \in M_x] = \varepsilon$, for some fixed $\varepsilon \in (0,1)$. We write
        \[\ket{\psi_x} = \frac{1}{\sqrt{\varepsilon}}\sum_{v \in M_x}\sqrt{\pi_v}\ket{v},\]
        and we define the state-reflection problems $R = \{(\ket{\perp} + \ket{\psi_x}, \ket{\perp} - \ket{\psi_x}, O_x)\}_{x \in \D}$ and $R' = \{(\mathbbm{1}_{\perp} - \frac{\pi|_{M_x}}{\varepsilon}, \mathbbm{1}_{\perp} + \mathbbm{1}_{M_x}, O_x)\}_{x \in \D}$. Then, the feasible regions for $R$ and $R'$ are identical, and $R'$ is a hyperedge problem.
    \end{theorem}

    \begin{proof}
        We use \cref{lem:rescale-state-reflection-problem} with a diagonal operator $D$ that acts as $-\diag(\sqrt{\pi})/\sqrt{\varepsilon}$ on $\C^V$, and as identity on $\ket{\perp}$, to argue that the feasible regions are identical. Next, checking that $R'$ is a hyperedge problem is done by directly checking the conditions from \cref{def:hyperedge-problem}.
    \end{proof}

    \subsection{Composition of hyperedges}

    The core idea of the generalized graph composition framework is to take a hypergraph $G = (V,E)$, where we divide the vertices into two groups, the boundary vertices $B \subseteq V$ and the internal vertices $V \setminus B$. Now, we associate a hyperedge problem to each of the hyperedges, in such a way that the net flow vanishes at every internal node, and such that the potential functions are uniquely defined at all vertices. We formalize this notion in the following definition.

    \begin{definition}[Generalized graph composition]
        \label{def:generalized-graph-composition}
        Let $G = (V,E)$ be a hypergraph, with boundary vertices $B \subseteq V$, and let $\D$ be a finite domain. For every hyperedge $e \in E$, let $R^e = \{(\delta_x^e, U_x^e, O_x^e)\}_{x \in \D}$ be a hyperedge problem. For all $v \in  V$ and $x \in \D$, we write $\delta_x(v) = \sum_{e \in N(v)} \delta_x^e(v)$. Suppose that
        \begin{enumerate}[nosep]
            \item For all internal nodes $v \in V \setminus B$, the net-flow vanishes, i.e., $\delta_x(v) = 0$.
            \item The potential function is uniquely defined at every vertex $v \in V$, i.e., $U_x^e(v) = U_x(v)$, for all $e \in N(v)$.
        \end{enumerate}
        Then, we say that $(G,B,(R^e)_{e \in E})$ is an instance of the generalized graph composition framework, that implements the hyperedge problem $R = \{(\delta_x, U_x, \bigoplus_{e \in E} O_x^e)\}_{x \in \D}$ on vertex set $B$.
    \end{definition}

    The framework now describes how we can combine feasible solutions for the individual hyperedge problems into a feasible solution for the hyperedge problem on the boundary vertices. This is the objective of the following theorem.

    \begin{theorem}[Composition of hyperedges]
        \label{thm:generalized-graph-composition}
        Let $(G = (V,E),B,(R^e)_{e \in E})$ be an instance of the generalized graph composition, implementing the hyperedge problem $R$. For all $e \in E$, suppose that $w^e = \{\ket{(w^e)_x^{\pm}}\}_{x \in \D}$ is a feasible solution to $R^e$. For all $x \in \D$, let
        \[\ket{w_x^{\pm}} := \bigoplus_{e \in E} \ket{(w^e)_x^{\pm}}.\]
        Then, $w = \{\ket{w_x^{\pm}}\}_{x \in \D}$ forms a feasible solution to $R$, with $\mathsf{R}_x^{\pm}(w) = \sum_{e \in E} \mathsf{R}_x^{\pm}(w^e)$, for all $x \in \D$.
    \end{theorem}

    \begin{proof}
        We check that the witnesses indeed form a feasible solution. To that end, recall that we only have to check the cross-talk between $+$ and $-$ states for state-reflection problems. So, let $x,y \in \D$, and we observe that
        \begin{align*}
            (\delta_x)|_B^T(U_y)|_B &= \delta_x^TU_y = \sum_{v \in V} \delta_x(v)U_y(v) = \sum_{\substack{e \in E \\ v \in N(e)}} \delta_x^e(v)U_y^e(v) = \sum_{e \in E} (\delta_x^e)^TU_y^e \\
            &= \sum_{e \in E} \bra{(w_x^+)^e}(I - (O_x^e)^{\dagger}O_y^e)\ket{(w_y^-)^e} = \bra{w_x^+}(I - O_x^{\dagger}O_y)\ket{w_y^-},
        \end{align*}
        where we wrote $O_x = \bigoplus_{e \in E} O_x^e$, for all $x \in \D$.
    \end{proof}

    We proceed by showing how this subsumes the ``simple'' graph composition framework, i.e., the construction from \cite[Definition~4.5]{cornelissen2025quantum}.

    \begin{theorem}[Generalization of simple graph composition]
        \label{thm:graph-composition}
        Every instance of graph composition can be turned into an instance for generalized graph composition, with the same witness sizes.
    \end{theorem}

    \begin{proof}
        Let $(G = (V,E))$ be an undirected graph with distinct source and sink nodes $s,t \in V$, and with edge span programs $\mathcal{P}^e$ on $\D$, for every $e \in E$. We take witness vectors $\{\ket{(w^e)_x^{\pm}}\}_{x \in \D}$ to each of the span programs $\mathcal{P}^e$. For all $x \in \D$, let $G(x) = (V,E(x))$ be the graph that contains all the edges $e \in E$ for which $x$ is positive for $\mathcal{P}^e$, with resistance function $r : e \mapsto \mathsf{R}_x^+(w^e)$.

        On the one hand, for every input $x \in \D$ for which $s$ and $t$ are connected in $G(x)$, let $(f_x^e)_{e \in E(x)}$ be the minimum-energy unit $st$-flow in $G(x)$, and let $f_x^e = 0$ for all $e \not\in E(x)$. Simultaneously, let $U_x : V \to \R$ be identically $0$ everywhere.

        On the other hand, for every $x \in \D$ for which $s$ and $t$ are not connected in $G(x)$, consider the graph $\overline{G}(x)$, where all the connected components in $G(x)$ are contracted, and the remaining edges from $G$ are present. Let the resistance function be $r : e \mapsto \mathsf{R}_x^-(w^e)^{-1}$. Let $U_x : V \to \R$ be the maximum-energy potential function in $\overline{G}(x)$ that satisfies $U_s - U_t = 1$, and let $f_x : E \to \R$ be the identically-zero function.

        Now, for every $uv = e \in E$, we define the hyperedge problem $R^e = \{(f_x^e(\mathbbm{1}_u - \mathbbm{1}_v), U_x|_{\{u,v\}}, 2\Pi_{\H^e(x)} - I)\}_{x \in \D}$, and we observe that it indeed satisfies the conditions from \cref{def:hyperedge-problem}. Moreover, we observe that $(G, \{s,t\}, (R^e)_{e \in E})$ is an instance of the generalized graph composition framework, as it satisfies all the conditions from \cref{thm:generalized-graph-composition}, and it implements the hyperedge problem that sends unit flow from $s$ to $t$ if $x$ is a positive input to the graph composition span program, and that has unit potential difference between $s$ and $t$ if not.

        Finally, we can use \cref{lem:rescale-state-reflection-problem} to obtain feasible solutions to the hyperedge problems $R^e$. Then, we use \cref{thm:generalized-graph-composition} to conclude that the total witness sizes become
        \[\mathsf{R}_x^+(w) = \sum_{e \in E} |f_e^x|^2 \cdot \mathsf{R}_x^+(w^e) = R_{\mathrm{eff}}(G(x), e \mapsto \mathsf{R}_x^+(w^e); s \leftrightarrow t),\]
        and
        \[\mathsf{R}_x^-(w) = \sum_{uv = e \in E} \left|U_u - U_v\right|^2 \cdot \mathsf{R}_x^-(w^e) = \sum_{uv = e \in E} \left|\frac{U_u - U_v}{\mathsf{R}_x^-(w^e)^{-1}}\right|^2 \cdot \mathsf{R}_x^-(w^e)^{-1} = R_{\mathrm{eff}}(\overline{G}(x), e \mapsto \mathsf{R}_x^-(w^e)^{-1}, s \leftrightarrow t)^{-1},\]
        where we used \cite[Theorem~4.4, item~3]{cornelissen2025quantum} to recover the witness sizes from \cite[Theorem~4.2]{cornelissen2025quantum}.
    \end{proof}

    Finally, similar to \cite[Theorem~4.3]{cornelissen2025quantum}, we give a somewhat less general formulation of the framework that is a bit easier to use. We refer to this result as the resistance-cut theorem.

    \begin{theorem}[Resistance-cut theorem for function evaluation and database update]
        \label{thm:path-cut-theorem}
        Let $G = (V,E)$ be a hypergraph, with boundary vertices $B \subseteq V$ and hyperedge problems $R^e = \{(\delta_x^e, U_x^e, O_x^e)\}_{x \in \D}$ for all $e \in E$. Suppose that for every $x \in \D$ and $e \in E$, the total flow vector $\delta_x^e$ has support on at most two vertices, and $U_x^e = 1$ for those vertices, and $0$ everywhere else. Then, for every $x \in \D$, we can make a simple graph $G(x) = (V,E(x))$, where we include all the edges $uv$ that are connected by the unit flow in the hyperedge problem. Suppose that $G(x)$ connects exactly two boundary vertices $\{s_x,t_x\}$. Then, for all $x \in \D$, let $C_x \subseteq E$ be a cut that certifies that all the other boundary vertices cannot  be reached from $s_x$ or $t_x$. Then, we can generate an instance to generalized graph composition with witness complexities
        \[\mathsf{R}_x^+(w) = R_{\mathrm{eff}}(G(x), e \mapsto \mathsf{R}_x^+(w^e), s_x \leftrightarrow t_x), \qquad \text{and} \qquad \mathsf{R}_x^-(w) = \sum_{e \in C_x} \mathsf{R}_x^-(w^e),\]
        where all $w^e$'s are feasible solutions to the $R^e$'s.
    \end{theorem}

    \begin{proof}
        We send the optimal flow from $s_x$ to $t_x$ over the graph $G(x)$, weighted by the witness sizes $\mathsf{R}_x^+(w^e)$, for all $e \in E$. The resulting complexity, similarly as in \cref{thm:graph-composition}, becomes the effective resistance between these two nodes in $G(x)$.

        For the negative witness size, we explicitly define the potential function to be $1$ on all the nodes that can be reached from $s_x$ and $t_x$ without taking a direction in a hyperedge in $C_x$ that is not connected by the flow at that edge. We take the potential function to be $0$ everywhere else. The witness size then simply becomes the sum over all the negative witness sizes of the hyperedges in $C_x$.
    \end{proof}

    We remark here that the resistance-cut theorem can be used to recover a path-cut theorem, similar to the one from \cite[Theorem~4.3]{cornelissen2025quantum}, by upper bounding the effective resistance by the total resistance along a single path from $s_x$ to $t_x$. We refer to this as the path-cut theorem for generalized graph composition.

    \subsection{Direct applications}
    \label{subsec:examples}

    We first give some straightforward applications of the path-cut theorem. To that end, we start with the problem that finds a first marked index in a bit string.

    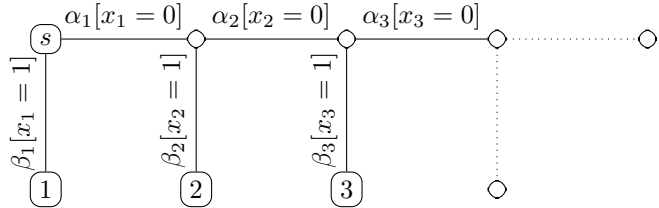
\begin{figure}[!ht]
        \centering
        \begin{tikzpicture}[vertex/.style = {draw, rounded corners = .4em}, yscale = -1, scale = 2]
            \node[vertex] (s) at (0,0) {$s$};
            \node[vertex] (1) at (0,1) {$1$};
            \node[vertex] (v2) at (1,0) {};
            \node[vertex] (2) at (1,1) {$2$};
            \node[vertex] (v3) at (2,0) {};
            \node[vertex] (3) at (2,1) {$3$};
            \node[vertex] (v4) at (3,0) {};
            \node[vertex] (4) at (3,1) {};
            \node[vertex] (v5) at (4,0) {};

            \draw (s) to node[above] {$\alpha_1[x_1 = 0]$} (v2) to node[above] {$\alpha_2[x_2 = 0]$} (v3) to node[above] {$\alpha_3[x_3 = 0]$} (v4);
            \draw[dotted] (v4) to (v5);
            \draw (s) to node[above, rotate = 90] {$\beta_1[x_1 = 1]$} (1);
            \draw (v2) to node[above, rotate = 90] {$\beta_2[x_2 = 1]$} (2);
            \draw (v3) to node[above, rotate = 90] {$\beta_3[x_3 = 1]$} (3);
            \draw[dotted] (v4) to (4);
        \end{tikzpicture}
        \caption{A generalized graph composition for the first marked index finding problem. The span programs $[x_j = b]$, for all $j \in [n]$ and $b \in \{0,1\}$ are trivial span programs that simply make one query, and for all $j \in [n]$, $\alpha_j,\beta_j > 0$.}
        \label{fig:first-marked-index-finding}
    \end{figure}

    \begin{theorem}
        Let $x \in \{0,1\}^n \setminus \{0^n\}$, and let $i(x)$ be the first index containing a $1$, i.e., $i(x) = \min\{i \in [n] : x_i = 1\}$. The graph composition construction from \cref{fig:first-marked-index-finding} computes $i(x)$ with complexities
        \[\mathsf{R}_x^+(w) = \sum_{j=1}^{i(x) - 1} \alpha_j + \beta_{i(x)}, \qquad \text{and} \qquad \mathsf{R}_x^-(w) = \sum_{j=1}^{i(x)-1} \frac{1}{\beta_j} + \frac{1}{\alpha_{i(x)}}.\]
        Both become $O(\sqrt{i(x)})$, if we choose $\alpha_j = 1/\sqrt{j}$ and $\beta_j = \sqrt{j}$, for all $j \in [n]$, and we have $\mathsf{R}_x^+(w) \in O(\log(i(x)))$ and $\mathsf{R}_x^-(w) \in O(i(x))$, if we choose $\alpha_j = 1/j$ and $\beta_j = 1$, for all $j \in [n]$.
    \end{theorem}

    \begin{proof}
        The complexities follow directly from \cref{thm:path-cut-theorem}. The path going from $s$ to $i(x)$ is unique, and for the positive witness size we simply sum over all the weights on the edges. For the negative witness size, we cut all the edges adjacent to the leaves labeled $\ell$ with $\ell < i(x)$, and we cut through the edge with span program $\alpha_{i(x)}[x_{i(x)} = 0]$.
    \end{proof}

    From the above construction, we observe that there are two weighting schemes with different qualitative behaviors. Indeed, by choosing $\alpha_j = 1/\beta_j = 1/\sqrt{j}$, we obtain that both the negative and positive complexities become $O(\sqrt{i(x)})$, which means that we can compute the first marked index in $O(\sqrt{i(x)})$ queries. However, if we wanted to have all the weight pushed onto one side, i.e., we wanted the positive witness sizes to be $1$ and the negative witness sizes to be $O(i(x))$, then the naive weighting scheme $\alpha_j = 1/j$ and $\beta_j = 1$ doesn't quite achieve that. Instead, we get the negative witness size to be $O(i(x))$, but the positive witness size picks up a logarithmic factor $O(\log(i(x)))$. We remark here that this behavior is fundamental, since a first marked index finding routine with constant positive witness size and negative witness size linear in $i(x)$ would imply a quantum algorithm with $O(\sqrt{T})$ queries for the $\mathrm{OR} \circ \mathrm{pSEARCH}$ problem, as defined in \cite[Section~4]{ambainis2023improved}, contradicting \cite[Theorem~4.2]{ambainis2023improved}.

    Next, we can learn an $n$-bit string using positive and negative witness complexity $n$.

    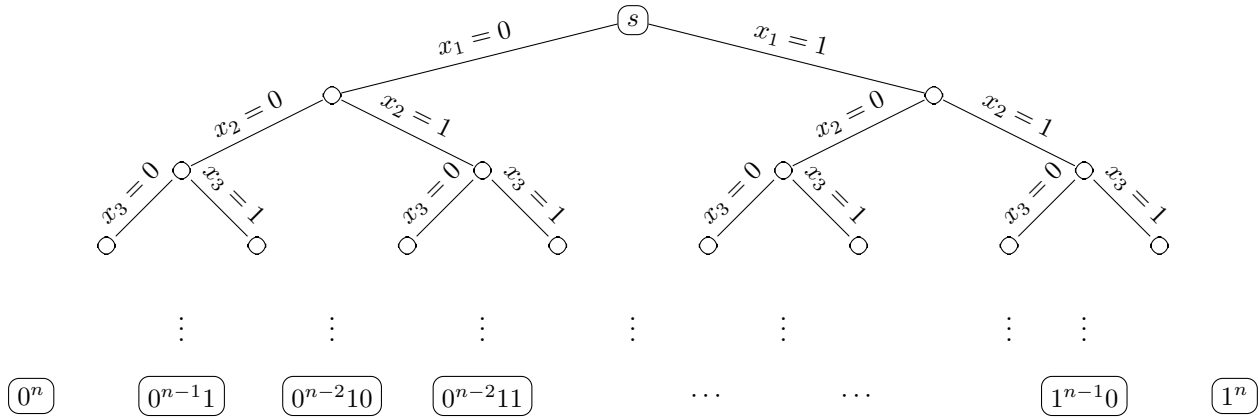
\begin{figure}[!ht]
        \centering
        \begin{tikzpicture}[vertex/.style = {draw, rounded corners = .4em}]
            \node[vertex] (s) at (0,0) {$s$};
            \node[vertex] (0) at (-4,-1) {};
            \node[vertex] (1) at (4,-1) {};
            \node[vertex] (00) at (-6,-2) {};
            \node[vertex] (01) at (-2,-2) {};
            \node[vertex] (10) at (2,-2) {};
            \node[vertex] (11) at (6,-2) {};
            \node[vertex] (000) at (-7,-3) {};
            \node[vertex] (001) at (-5,-3) {};
            \node[vertex] (010) at (-3,-3) {};
            \node[vertex] (011) at (-1,-3) {};
            \node[vertex] (100) at (1,-3) {};
            \node[vertex] (101) at (3,-3) {};
            \node[vertex] (110) at (5,-3) {};
            \node[vertex] (111) at (7,-3) {};

            \node[vertex] (0n) at (-8,-5) {$0^n$};
            \node[vertex] (0n1) at (-6,-5) {$0^{n-1}1$};
            \node[vertex] (0n10) at (-4,-5) {$0^{n-2}10$};
            \node[vertex] (0n11) at (-2,-5) {$0^{n-2}11$};
            \node at (1,-5) {$\cdots$};
            \node at (3,-5) {$\cdots$};
            \node[vertex] (1n0) at (6,-5) {$1^{n-1}0$};
            \node[vertex] (1n) at (8,-5) {$1^n$};

            \node at (-6,-4) {$\vdots$};
            \node at (-4,-4) {$\vdots$};
            \node at (-2,-4) {$\vdots$};
            \node at (0,-4) {$\vdots$};
            \node at (2,-4) {$\vdots$};
            \node at (5,-4) {$\vdots$};
            \node at (6,-4) {$\vdots$};

            \draw (s) to node[above, rotate = {atan(1/4)}] {$x_1 = 0$} (0);
            \draw (s) to node[above, rotate = {-atan(1/4)}] {$x_1 = 1$} (1);

            \draw (0) to node[above, rotate = {atan(1/2)}] {$x_2 = 0$} (00);
            \draw (0) to node[above, rotate = {-atan(1/2)}] {$x_2 = 1$} (01);
            \draw (1) to node[above, rotate = {atan(1/2)}] {$x_2 = 0$} (10);
            \draw (1) to node[above, rotate = {-atan(1/2)}] {$x_2 = 1$} (11);

            \draw (00) to node[above, rotate = 45] {$x_3 = 0$} (000);
            \draw (00) to node[above, rotate = -45] {$x_3 = 1$} (001);
            \draw (01) to node[above, rotate = 45] {$x_3 = 0$} (010);
            \draw (01) to node[above, rotate = -45] {$x_3 = 1$} (011);
            \draw (10) to node[above, rotate = 45] {$x_3 = 0$} (100);
            \draw (10) to node[above, rotate = -45] {$x_3 = 1$} (101);
            \draw (11) to node[above, rotate = 45] {$x_3 = 0$} (110);
            \draw (11) to node[above, rotate = -45] {$x_3 = 1$} (111);
        \end{tikzpicture}
        \caption{A generalized graph composition construction for learning an $n$-bit string.}
        \label{fig:dense-learning}
    \end{figure}

    \begin{theorem}
        The graph composition construction from \cref{fig:dense-learning} admits a feasible solution $w$ that learns a bit string $x \in \{0,1\}^n$ with witness complexities
        \[\mathsf{R}_x^+(w) = \mathsf{R}_x^-(w) = n.\]
    \end{theorem}

    \begin{proof}
        The proof follows from \cref{thm:path-cut-theorem}. Every path from root to leaf has $n$ edges, so the positive witness complexity is $n$. Similarly, for every node visited along the path, there is exactly one edge that is not available. We group all those in a cut, and we obtain a cut with exactly $n$ edges as well.
    \end{proof}

    We remark here that it is possible to slightly modify the above construction, to improve the complexities in the case where the Hamming weight of the bit string is small, akin to \cite[Lemma~4]{carette2020extended}. We leave porting these ideas to the graph composition framework for future work.

    Finally, we note that we can also recover quantum minimum-finding in the generalized graph composition framework, albeit with a polylogarithmic slowdown.

    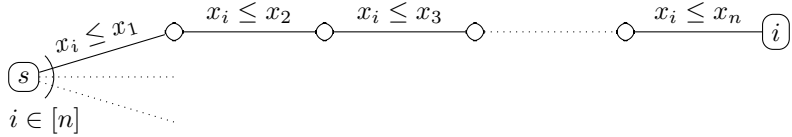
\begin{figure}[!ht]
        \centering
        \begin{tikzpicture}[vertex/.style = {draw, rounded corners = .4em}, scale = 2]
            \node[vertex] (s) at (0,0) {$s$};
            \node[vertex] (1) at (1,.3) {};
            \node[vertex] (2) at (2,.3) {};
            \node[vertex] (3) at (3,.3) {};
            \node[vertex] (4) at (4,.3) {};
            \node[vertex] (i) at (5,.3) {$i$};

            \draw (s) to node[above, rotate = {atan(.3)}] {$x_i \leq x_1$} (1);
            \draw (1) to node[above] {$x_i \leq x_2$} (2);
            \draw (2) to node[above] {$x_i \leq x_3$} (3);
            \draw[dotted] (3) to (4);
            \draw (4) to node[above] {$x_i \leq x_n$} (i);

            \draw ({.2*cos(45)},{.2*sin(45)}) arc(45:-45:.2) node[below] {$i \in [n]$};
            \draw[dotted] (s) to (1,0);
            \draw[dotted] (s) to (1,-.3);
        \end{tikzpicture}
        \caption{A generalized graph composition construction for quantum minimum finding.}
        \label{fig:minimum-finding}
    \end{figure}

    \begin{theorem}
        The generalized graph composition construction from \cref{fig:minimum-finding} solves the minimum finding problem with complexities
        \[\mathsf{R}_x^+ \in O(n), \qquad \text{and} \qquad \mathsf{R}_x^- \in O(\log(n)).\]
    \end{theorem}

    \begin{proof}
        We use \cref{thm:generalized-graph-composition} directly. We observe that there is a single path consisting of $n$ edges connecting $s$ with the index of the minimum element $i$, and so by sending unit flow across those edges, the positive witness size is $n$. For the negative witness size, we observe that we have to cut all the other connections between $s$ and labels $\ell \in [n]$ for which $\ell \neq i$. To that end, observe that for the $j$th smallest element in the list, there are $j-1$ edges that can be cut in the connection between $s$ and its corresponding leaf. Thus, we observe that the negative witness size is
        \[\mathsf{R}_x^-(w) = \sum_{j=2}^n \frac{1}{j-1} \in O(\log(n)).\qedhere\]
    \end{proof}

    It is a little unsatisfying that we cannot recover D\"urr and H{\o}yer's minimum finding routine exactly \cite{durr1996quantum}, since they show how this problem can be solved in $O(\sqrt{n})$ comparison queries. It would be very instructive to figure out if the log-factor can be removed here with a more intricate construction, but we leave those considerations for future work.

    \section{Graph composition for non-boolean decision trees}
    \label{sec:decision-trees}

    In this section, we discuss how decision-trees can be turned into quantum algorithms. To that end, suppose that we have a decision tree that computes a function $f : \D \to \Sigma_O$, where $\D \subseteq \Sigma_I^n$ for some input and output alphabets $\Sigma_I$ and $\Sigma_O$.

    The first result in this direction was obtained by Lin and Lin~\cite{lin2016upper}, who considered the setting where $\Sigma_I = \Sigma_O = \{0,1\}$. They showed that a decision tree with depth $T$ can be turned into a quantum query algorithm, where the number of queries depends on the structure of the tree. Specifically, color the edges of the decision tree in red and black, such that every node has at most one outgoing black edge. Let $G$ be the maximum number of red edges along a path from the root to a leaf. Then, the total number of queries is $O(\sqrt{GT})$. This result was later generalized by \cite{beigi2020quantum} to the setting where $\Sigma_I$ and $\Sigma_O$ are arbitrary.

    Next, this question was revisited in \cite{cornelissen2025improved}. They showed that in the case where $\Sigma_I = \{0,1\}$, one can assign a weighting scheme on the edges, that generalizes the coloring scheme from the earlier results. After that, it was shown in \cite{cornelissen2025quantum} how to recover this result, in the case where $\Sigma_I = \Sigma_O = \{0,1\}$ in from the graph composition framework. In both of these works, the question is left open whether such a weighting scheme can be imposed on decision trees with $\Sigma_I \neq \{0,1\}$.

    Here, we employ the generalized graph composition framework to convert decision trees with non-binary inputs and outputs into quantum algorithms. We observe that this naturally gives rise to a weighting scheme, where the optimal assignment can be obtained by recursively solving a simple semi-definite program at every vertex. We also show that this can be used to recover the results from \cite{beigi2020quantum,cornelissen2025improved}.

    We start by defining a weighting scheme on decision trees.

    \begin{definition}
        \label{def:weighting-scheme}
        Let $T = (V,E)$ be a decision tree with root node $r \in V$, and $w : V \to \R_{\geq0}$. We say that $w$ is a weighting scheme for $T$, if:
        \begin{enumerate}[nosep]
            \item $w_v = 0$ for all leafs $v$.
            \item For any internal node $v$ with children $C_v \subseteq V$, we have matrices $0 \preceq X,Y \in \C^{C_v \times C_v}$ such that $X[c_1,c_2] - Y[c_1,c_2] = 1$, for all $c_1,c_2 \in C_V$ such that $c_1 \neq c_2$, and $w_v - w_c \geq X[c,c] + Y[c,c]$ for all $c \in C_v$.
        \end{enumerate}
        We say that $w_r$ is the weight of the tree, and we write
        \[\mathsf{WDT}(T) = \min_{w \text{ weighting scheme for } T} w_r.\]
    \end{definition}

    We start by showing that these weighting schemes provide upper bounds on the query complexity of the function being computed by the decision tree.

    \begin{theorem}
        Let $T = (V,E)$ be a decision tree with root node $r \in V$ and weighting scheme $w$, computing a function $f : \Sigma_I^n \supseteq \D \to \Sigma_O$. Then, we can compute $f$ with bounded error using $O(w_r)$ queries. Consequently, $\mathsf{Q}(f) \in O(\mathsf{WDT}(f))$.
    \end{theorem}

    \begin{proof}
        We use the generalized graph composition framework. To that end, we let $G = (V,E')$ be a hypergraph, such that for every internal node $v \in V$ of the decision tree, with children $C_v$, we let $\{v\} \cup C_v$ be a hyperedge in $G$. Now, we associate the function-evaluation hyperedge problem to it, that on an input $x$ whose path from root to leaf reaches $v$, sends flow $1$ from the parent node to its child that matches the query outcome, and assigns potential function $1$ to these two nodes, and $0$ everywhere else.

        We now present a feasible solution to the hyperedge problem associated to the hyperedge with parent node $v$ and children $C_v$. As $X$ and $Y$ are PSD, we can factorize $X = U_+^{\dagger}U$ and $Y = U_-^{\dagger}U$, where for every $c \in C_v$ we write $\ket{u_c^+}$ and $\ket{u_c^-}$ for the $c$th column in $U_+$ and $U_-$, respectively. Now, to every input $x \in \D$ that passes through node $v$ on the path from root to leaf in $T$, we associate the vectors $\ket{w_x^{\pm}} = (\ket{u^{\pm}_{x_v}} \oplus \pm\ket{u^{\mp}_{x_v}}) \otimes (\ket{\perp} \pm \ket{x_v})$, and we set $\ket{w_x^{\pm}} = 0$ otherwise. For all $x,y \in \D$ whose paths pass through $v$, we let $x_v$ and $y_v$ be the value queried at node $v$, and observe that
        \[\braket{w_x^+}{w_y^-} = \left(\braket{u_{x_v}^+}{u^-_{y_v}} - \braket{u_{x_v}^-}{u^+_{y_v}}\right) \cdot (1 - \braket{x_v}{y_v}) = (X[x_v,y_v] - Y[x_v,y_v]) \cdot \mathbbm{1}_{x_v \neq y_v} = \mathbbm{1}_{x_v \neq y_v} = \delta_x^TU_y,\]
        and in all other cases, the inner product is $0$.`

        Now, we combine all the feasible solutions using \cref{thm:generalized-graph-composition}, with all weights equal to $1$. The combined flow becomes simply from the root node of the graph to the leaf computed by the decision tree, and the potential function is $1$ on these two vertices and $0$ everywhere else. This is the function-evaluation state-reflection problem, so it remains to compute the witness sizes. To that end, for any input $x \in \D$, let $P_x$ be the hyperedges that are traversed upon evaluating $x$, and $p(e)$ the corresponding parent nodes. We find that
        \begin{align*}
            \mathsf{R}_x^{\pm}(w) &= \sum_{e \in P_x} (\mathsf{R}_e)_x^{\pm} = \sum_{e \in P_x} \norm{\ket{(w_e)_x^{\pm}}}^2 = 2 \sum_{e \in P_x} \left(\norm{\ket{(u^e)^+_{x_{p(e)}}}}^2 + \norm{\ket{(u^e)^-_{x_{p(e)}}}}^2\right) \\
            &= 2\sum_{e \in P_x} (X^e[x_{p(e)},x_{p(e)}] + Y^e[x_{p(e)},x_{p(e)}]) \leq 2\sum_{e \in P_x} (w_{p(e)} - w_{x_{p(e)}}) = 2w_r = 2\mathsf{WDT}(T),
        \end{align*}
        and so the total complexity is upper bounded by $2\mathsf{WDT}(T)$.
    \end{proof}

    Next, we explain how the optimal weighting scheme can be recursively computed.

    \begin{theorem}
        \label{thm:optimal-weight-assignment}
        Let $T = (V,E)$ be a decision tree with root node $r \in V$. Its optimal weighting scheme is computed by setting all the weight of all the leaf nodes to $0$, and then for every internal node $v \in V$ with children $C_v \subseteq V$, recursively computing the semi-definite program
        \begin{align*}
            w_v := \min\quad & \max_{c \in C_v} (w_c + X[c,c] + Y[c,c]) \\
            \text{s.t.}\quad & X[c_1,c_2] - Y[c_1,c_2] = 1, & \text{if } c_1,c_2 \in C_v, c_1 \neq c_2, \\
            & 0 \preceq X,Y \in \R^{C_v \times C_v}.
        \end{align*}
    \end{theorem}

    \begin{proof}
        We immediately observe that the weighting scheme generated in this way is feasible for $T$, since the constraints of the SDP are the same as those in \cref{def:weighting-scheme}. Similarly, every valid weighting scheme for $T$ indeed generates feasible solutions to all the SDPs. Finally, we observe that the objective function is monotone in the values of $w_c$, and so the value of $w_r$ is optimized if and only if all of the intermediate SDPs are optimized.
    \end{proof}

    For future reference we also provide the dual of the SDP from \cref{thm:optimal-weight-assignment}. For all internal nodes $v \in V$ with children $C_v \subseteq V$, we have
    \begin{align*}
        w_v := \max \quad &\norm{\diag(\{w_c : c \in C_v\}) + \Gamma}, \\
        \text{s.t.} \quad &\norm{\Gamma} \leq 1, \\
        &\Gamma[c,c] = 0, & \forall c \in C_v, \\
        &\Gamma \in \R^{C_v \times C_v}.
    \end{align*}

    We can now obtain the result from \cite{beigi2020quantum} by exhibiting an explicit weighting scheme. This is the objective of the following theorem.

    \begin{theorem}
        \label{thm:BT20}
        For every decision tree $T = (V,E)$, we have $\mathsf{WDT}(T) \leq 3\sqrt{GT}$.
    \end{theorem}

    \begin{proof}
        We exhibit a weighting scheme. We set the weights such that for every black edge $vv'$, we have $w_v - w_{v'} \geq \sqrt{G/T}$, and similarly for every red edge $w_v - w_{v'} \geq 2\sqrt{T/G} - 1$. We can indeed do this in a way that ensures $w_r \leq 3\sqrt{GT}$. It remains to prove that such a weighting scheme is feasible. To that end, let $k \in \N$ be the number of outcomes of the query. We let $\alpha = \sqrt{T/G} \geq 1$, $v = [1/\sqrt{\alpha}] + [\sqrt{\alpha}]^{k-1}$, and
        \[X := vv^T \succeq 0, \qquad \text{and} \qquad Y := \begin{bmatrix}
            0 & ([0]^{k-1})^T \\
            [0]^{k-1} & (\alpha-1)J_{k-1 \times k-1}
        \end{bmatrix} \succeq 0.\]
        We verify for all $c,c' \in C$ with $c \neq c'$ satisfies $X[c,c'] - Y[c,c'] = 1$, $X[1,1] + Y[1,1] = 1/\alpha$, and $X[c,c] + Y[c,c] = 2\alpha - 1$.
    \end{proof}

    Finally, we recover the results from \cite{cornelissen2025improved}, by observing that in the case where the query has just two outcomes, the corresponding SDP can be solved analytically.

    \begin{theorem}
        \label{thm:CMP25}
        For a decision tree with binary queries, the measure $\mathsf{WDT}(T)$ introduced in \cref{def:weighting-scheme} equals the optimum of \cite[Definition~1.5]{cornelissen2025improved}.
    \end{theorem}

    \begin{proof}
        Let $v \in V$ be an internal node, with weights $w_v$ and $w_{v'}$ on its children. Then, the optimal choice for the $\Gamma$-matrix in the maximization version of the SDP is the Pauli-$X$ matrix, from which we observe that the optimal choice for $w_v$ is
        \[w_v = \norm{\begin{bmatrix}
            w_v & 1 \\
            1 & w_{v'}
        \end{bmatrix}} = \frac{w_v + w_{v'} + \sqrt{(w_v - w_{v'})^2 + 4}}{2}.\]
        This recovers \cite[Theorem~5.4]{cornelissen2025improved}.
    \end{proof}

    For the non-Boolean case, finding an explicit analytical solution for this optimization program seems elusive. We do remark, however, that these SDPs are very small (only of the size of the input alphabet), and that the total number of SDPs to solve is the same as the size of the tree. Thus, for input alphabets of constant size, computing the optimal weighting schemes can be done in time essentially linear in the size of the decision tree.

    \section{Graph composition for quantum divide and conquer}
    \label{sec:divide-and-conquer}

    In this section, we relate our construction to the results from \cite{childs2025quantum}. In \cite{cornelissen2025quantum}, it is already shown how one can recover Strategy~1 from this paper using the regular graph composition framework. here, we show how Strategy~2 can be obtained from the generalized graph composition framework.

    \begin{theorem}
        \label{thm:divide-and-conquer}
        Let $f^{\mathrm{aux}} : \Sigma^n \to \Lambda$ be an auxiliary function, and for all $s \in \Lambda$ and $i \in A$ with $A$ a finite set, let $f_i^{(s)} : \Sigma^n \to \{0,1\}$ and let $g^{(s)} : \{0,1\}^{\Lambda \times A} \to \{0,1\}$. Then, for any $x \in \Sigma^n$, we can compute $h^{f^{\mathrm{aux}}(x)}(x) := g^{(f^{\mathrm{aux}}(x))}((f_i^{(f^{\mathrm{aux}}(x))}(x))_{i \in A})$, with a feasible solution that satisfies
        \[\mathsf{R}_x^+(w) \leq (\mathsf{R}^{\mathrm{aux}})_x^+ + (\mathsf{R}^{f^{\mathrm{aux}}(x)})_x^+, \qquad \text{and} \qquad \mathsf{R}_x^-(w) = (\mathsf{R}^{\mathrm{aux}})_x^- + (\mathsf{R}^{f^{\mathrm{aux}}(x)})_x^-.\]
    \end{theorem}

    \begin{proof}
        The idea is to simply first perform the auxiliary hyperedge, and then attach all the hyperedges computing $h^{(s)}$ to the corresponding output nodes. The resulting witness analysis then follows from \cref{thm:generalized-graph-composition}, by sending unit flow through both the hyperedges, and the potential function is simply $0$ for all the nodes outside of this flow.
    \end{proof}

    The construction sketched here allows for more freedom than the one outlined in \cite{childs2025quantum}. We can even take a feasible solution to the SDP from the previous section, and multiply it with our feasible solution element-wise, to better balance the costs of the difference branches of computation, dependent on $s \in \Lambda$. We leave investigating the potential of this approach for future research.

    \section{Graph composition for quantum walk search}
    \label{sec:quantum-walk-search}

    In this section, we show how we can use the generalized graph composition framework to implement quantum walks. The idea of quantum walks originated in \cite{szegedy2004quantum}, in what became later known as the ``hitting time framework''. This idea was later built upon by several works, most notably \cite{magniez2011search}, \cite{belovs2013quantum}, and \cite{dohotaru2017controlled}, and these results were subsequently unified by \cite{apers2021unified}. Subsequently, Jeffery showed how some of these results can be amortized \cite{jeffery2022quantum}, and we build on those ideas here.

    \subsection{Construction}

    The idea of quantum walk search is as follows. Suppose we have a connected undirected graph $G = (V,E)$, with a resistance function on the edges $r : E \to \R_{>0}$. By \cref{thm:random-walk-characterization}, this defines an irreducible, reversible Markov process $M = (V,P)$.

    Now, let $\D$ be a finite domain, and for all $x \in \D$, let $M_x \subseteq V$ be the set of \textit{marked} vertices. The goal is to \textit{decide} whether $M_x \neq \varnothing$, or to \textit{find} an element $v \in M_x$, if $M_x$ is non-empty. We refer to these two problems as the \textit{detection} and \textit{finding} version of quantum walk search.

    Intuitively, we can think of an algorithm solving the detection and finding version of quantum walk search as follows. We select a vertex $v \in V$ according to some input distribution $\sigma \in \Delta_V$. Next, we randomly walk on the graph $G$ according to the dynamics prescribed by the random walk $M$, and once every so often, we check if the vertex $v \in V$ that we visit is contained in the set $M_x$. If it is, we have successfully solved both the detection and finding version of the problem, and if we can prove an upper bound on the time it takes to find such a marked vertex, if it exists, then we can also use this idea to verify that the set of marked vertices is empty, with some reasonable probability.

    In order to aid the step where we are checking whether $v \in M_x$, it can sometimes be beneficial to store some information about $x$ as we are walking on the vertices in $G$. To that end, we define a database with possible states $\mathcal{S}$, and we associate to every vertex $v \in V$ and input $x \in \D$, a database entry $D_{v,x} \in \mathcal{S}$. The idea is to compute the database entry exactly in the initial sampling step, and then to update the database entry at every step of the random walk. We can then use the database entry $D_{v,x}$ to check whether $v \in M_x$, or not.

    In order to implement this approach quantumly, we need to be able to perform three operations as subroutines. That is, we will require to have hyperedges that implement the following three operations:

    \begin{enumerate}[nosep]
        \item For all $v \in V$, we need a \textit{setup routine} $S_v$ that computes the database entry $D_{v,x}$, for all $x \in \D$.
        \item For all $vw = e \in E$, we need an \textit{update routine} that updates the database entry $D_{v,x}$ to $D_{w,x}$, for all $x \in \D$.
        \item For all $v \in V$ and $D \in \mathcal{S}$, we need a \textit{checking routine} that checks whether $v \in M_x$, for all $ \in \D$, if $D = D_{v,x}$.
    \end{enumerate}

    We bundle all these objects into a single \textit{instance of quantum walk search}, and we formally define it below.

    \begin{definition}[Quantum walk search]
        Let $G = (V,E)$ be a connected, undirected graph, with resistances $r : E \to \R_{>0}$, such that $\sum_{e \in E} r_e^{-1} = 1/2$. Let $\D$ be a finite domain, and for all $x \in \D$, let $M_x \subseteq V$. Let $\mathcal{S}$ be a set of database entries, and for all $v \in V$ and $x \in \D$, let $D_{v,x} \in \mathcal{S}$. Now, assume we have feasible solutions implementing the following operations:
        \begin{enumerate}[nosep]
            \item For all $v \in V$, let $S_v$ be a hyperedge that computes $D_{v,x}$, for all $x \in \D$. We write the positive and negative witness sizes as $(\mathsf{S}_v)_x^{\pm}$.
            \item For all $vw = e \in E$, let $\mathsf{U}_e$ be a hyperedge that updates $D_{v,x}$ to $D_{w,x}$, for all $ \in \D$. We write the positive and negative witness sizes as $(\mathsf{U}_e)_x^{\pm}$.
            \item For all $v \in V$ and $D \in \mathcal{S}$, let $C_{v,D}$ be a span program that checks whether $v \in M_x$, for all $x \in \D$, whenever $D = D_{v,x}$. We write the positive and negative witness sizes as $(\mathsf{C}_{v,D})_x^{\pm}$.
        \end{enumerate}
        Then, we say that $(G, r, x \mapsto M_x, (v,x) \mapsto D_{v,x}, \{S_v\}_{v \in V}, \{U_e\}_{e \in E}, \{C_{v,D}\}_{v \in V, D \in \mathcal{S}})$ is an instance of quantum walk search.
    \end{definition}

    We note that it is okay if $C_{v,D}$ errs if $D \neq D_{v,x}$, i.e., if the database entry $D$ is not the entry associated with the input $x$ at vertex $v$, then it is okay that the checking routine fails. The construction of quantum walk search will make sure that this doesn't impact the algorithm's correctness.

    We now show how we can solve the detection and finding versions of quantum random walk search, using these subroutines as hyperedges. We start with the detection version of the problem.

    \begin{theorem}[Graph composition for quantum walk search -- detection version]
        \label{thm:quantum-walk-detection}
        Let $(G, r, x \mapsto M_x, (v,x) \mapsto D_{v,x}, \{S_v\}_{v \in V}, \{U_e\}_{e \in E}, \{C_{v,D}\}_{v \in V, D \in \mathcal{S}})$ be an instance of quantum walk search. Let $\sigma,\tau \in \Delta_V$. We write $\D^+ = \{x \in \D : M_x \neq \varnothing\}$, and $\D^- = \D \setminus \D^+$. For all $x \in \D^+$, let $\mu_x \in \Delta_{\supp(\sigma)}$ and $\nu_x \in \Delta_{M_x \cap \supp(\tau)}$. Then, there is a span program $\mathcal{P}$, that solves the detection version of quantum walk search, with witness sizes
        \begin{align*}
            w_+(x,\mathcal{P}) &\leq \sum_{v \in \supp(\sigma)} \frac{(\mu_x)_v^2}{\sigma_v}(\mathsf{S}_v)_x^+ + R_{\mathrm{eff}}(G, r \circ (\mathsf{U}_{\cdot})_x^+; \mu_x - \nu_x) + \sum_{v \in \supp(\tau) \cap M_x} \frac{(\nu_x)_v^2}{\tau_v}(\mathsf{C}_{v,D_{v,x}})_x^+, & \forall x \in \D^+, \\
            w_-(x,\mathcal{P}) &\leq \sum_{v \in \supp(\sigma)} \sigma_v(\mathsf{S}_v)_x^- + \sum_{e \in E} \frac{(\mathsf{U}_e)_x^-}{r_e} + \sum_{v \in \supp(\tau)} \tau_v(\mathsf{C}_{v,D_{v,x}})_x^-, & \forall x \in \D^-.
        \end{align*}
    \end{theorem}

    \begin{proof}
        We perform a generalized graph composition, with the following graph structure. For all $v \in V$, we define $V_v = \{(v,D) : D \in \mathcal{S}\}$, and we use the new vertex set
        \[V' = \{s,t\} \cup \bigcup_{v \in V} V_v,\]
        with boundary vertices $B = \{s,t\}$. For all $v \in \supp(\sigma)$, we embed the hyperedge $S_v$ between vertices $\{s\} \cup V_v$ with weight $1/\sigma_v$, and for all $vw = e \in E$, we embed the hyperedge $U_e$ between vertices $V_v \cup V_w$ with weight $r_e$. Finally, for all $v \in \supp(\tau)$ and $D \in \mathcal{S}$, we convert $C_{v,D}$ into a hyperedge using \cref{thm:hyperedge-from-span-program}, and we embed it between vertices $\{(v,D)\} \cup \{t\}$ with weight $1/\tau_v$.

        We now construct a flow from $s$ to $t$ for $x \in \D^+$. To that end, we send flow $(\mu_x)_v$ through $S_v$, and $(\nu_x)_v$ through $C_{v,D_{v,x}}$. Thus, in order to satisfy flow conservation, we have to generate a net-flow of $\mu_x - \nu_v$ through the update edges. We choose the minimum-energy flow through the graph $G$, with the resistances computed by $e \mapsto r_e \cdot (\mathsf{U}_e)_x^+$. The corresponding positive witness complexity, hence, becomes
        \[\mathsf{R}_x^+ := \sum_{v \in \supp(\sigma)} \frac{(\mu_x)_v^2}{\sigma_v}(\mathsf{S}_v)_x^+ + R_{\mathrm{eff}}(G, r \circ (\mathsf{U}_{\cdot})_x^+; \mu_x - \nu_x) + \sum_{v \in \supp(\tau)} \frac{(\nu_v)_x^2}{\tau_v} (\mathsf{C}_{v,D_{v,x}})_x^+.\]

        Next, we construct a unit potential function $U_x$ for $x \in \D^-$, which is $1$ at $s$, and $0$ at $t$. We assign a potential of $1$ at every vertex in $V'$ that can be reached from $s$, and a potential of $0$ at all the others. Then, all the nodes of the form $(v,D_{v,x})$ and $s$ have potential function $1$, and all the others have $0$. Thus, all the setup and update hyperedges contribute to the negative witness complexity, and all the checking routines of the form $C_{v,D_{v,x}}$ contribute as well. This yields the negative witness complexity
        \[\mathsf{R}_x^- := \sum_{v \in \supp(\sigma)} \sigma_v(\mathsf{S}_v)_x^- + \sum_{e \in E} \frac{(\mathsf{U}_e)_x^-}{r_e} + \sum_{v \in \supp(\tau)} \tau_v(\mathsf{C}_{v,D_{v,x}})_x^-.\]

        Now, we know from \cref{thm:hyperedge-from-span-program} that we can turn the resulting hyperedge back into a span program, and since we have exhibited a feasible solution with witness complexities $\mathsf{R}_x^{\pm}$, we can upper bound its witness sizes accordingly, by taking the maximum of $x \in \D^{\pm}$, respectively.
    \end{proof}

    We can generalize this result to the finding case as well. The result carries over immediately when we have the promise that $|M_x| = 1$.

    \begin{theorem}[Graph composition for quantum walk search -- finding version for unique marked element]
        \label{thm:quantum-walk-finding-unique}
        Let $(G, r, x \mapsto M_x, (v,x) \mapsto D_{v,x}, \{S_v\}_{v \in V}, \{U_e\}_{e \in E}, \{C_{v,D}\}_{v \in V, D \in \mathcal{S}})$ be an instance of quantum walk search, with the promise that for all $x \in \D$, $M_x = \{m_x\}$. Let $\sigma,\tau \in \Delta_V$. For all $x \in \D$, let $\mu_x \in \Delta_{\supp(\sigma)}$. Then, we can find the marked vertex with complexities
        \begin{align*}
            \mathsf{R}_x^+ &\leq \sum_{v \in \supp(\sigma)} \frac{(\mu_x)_v^2}{\sigma_v}(\mathsf{S}_v)_x^+ + R_{\mathrm{eff}}(G, r \circ (\mathsf{U}_{\cdot})_x^+; \mu_x - \nu_x) +  \frac{1}{\tau_{m_x}}(\mathsf{C}_{m_x,D_{m_x,x}})_x^+, & \forall x \in \D, \\
            \mathsf{R}_x^- &\leq \sum_{v \in \supp(\sigma)} \sigma_v(\mathsf{S}_v)_x^- + \sum_{e \in E} \frac{(\mathsf{U}_e)_x^-}{r_e} + \sum_{v \in \supp(\tau) \setminus \{m_x\}} \tau_v(\mathsf{C}_{v,D_{v,x}})_x^-, & \forall x \in \D.
        \end{align*}
    \end{theorem}

    \begin{proof}
        This follows from \cref{thm:path-cut-theorem}. We use the same flow as in the detection case to upper bound the effective resistance, and we use the cut that cuts everything that can be reached from the root node from all the others. Since we are promised that there is a unique marked element, the resulting flow ends up connecting two vertices, which we can turn back into a state-reflection problem with unit-norm states using \cref{thm:hyperedge-from-database-function} in reverse.
    \end{proof}

    We also show how we can solve the finding problem in we know the fraction of marked vertices.

    \begin{theorem}[Graph composition for quantum walk search -- finding version for known fraction of marked vertices]
        \label{thm:quantum-walk-finding-fixed-fraction}
        Let $(G, r, x \mapsto M_x, (v,x) \mapsto D_{v,x}, \{S_v\}_{v \in V}, \{U_e\}_{e \in E}, \{C_{v,D}\}_{v \in V, D \in \mathcal{S}})$ be an instance of quantum walk search. Let $\sigma,\tau \in \Delta_V$, and suppose that for all $x \in \D$, we have $\P_{v \sim \tau}[v \in M_x] = \varepsilon$ for some fixed $\varepsilon > 0$. For all $x \in \D$, let $\mu_x \in \Delta_{\supp(\sigma)}$, and $\nu_x = (\tau|_{M_x})/\varepsilon$. Then, we can find the marked vertex with complexities
        \begin{align*}
            \mathrm{R}_x^+ &\leq \sum_{v \in \supp(\sigma)} \frac{(\mu_x)_v^2}{\sigma_v}(\mathsf{S}_v)_x^+ + R_{\mathrm{eff}}(G, r \circ (\mathsf{U}_{\cdot})_x^+; \mu_x - \nu_x) + \frac{1}{\varepsilon^2}\sum_{v \in \supp(\tau) \cap M_x} \tau_v(\mathsf{C}_{v,D_{v,x}})_x^+, & \forall x \in \D, \\
            \mathrm{R}_x^- &\leq \sum_{v \in \supp(\sigma)} \sigma_v(\mathsf{S}_v)_x^- + \sum_{e \in E} \frac{(\mathsf{U}_e)_x^-}{r_e} + \sum_{v \in \supp(\tau) \setminus M_x} \tau_v(\mathsf{C}_{v,D_{v,x}})_x^-, & \forall x \in \D.
        \end{align*}
    \end{theorem}

    \begin{proof}
        The idea is to use the same flow as in the detection version for the positive witness analysis. For the negative witness analysis, we use the same idea as in the previous theorem, i.e., we let the potential function be $1$ for all the vertices that we can reach from $s$, and $0$ everywhere else. Now, we end up with the same situation as in \cref{thm:known-marked-fraction-recovery}, and so we can again multiply with a diagonal to obtain a state-reflection problem with unit-norm states.
    \end{proof}

    It is at the moment not clear how this can be generalized to the setting where the fraction of marked vertices is unknown. We leave resolving this case for future research.

    \subsection{Relation to unified quantum walk search~\cite{apers2021unified}}

    We now relate our results to those obtained by Apers, Gily\'en and Jeffery~\cite{apers2021unified}. We first of all show that w.r.t.\ the query complexity, one can remove all logarithmic factors from the complexity statement.

    \begin{theorem}
        \label{thm:unified-quantum-walk-search}
        Let $(G,r,x \mapsto M_x,(v,x) \mapsto D_{v,x}, \{S_v\}_{v \in V}, \{U_e\}_{e \in E}, \{C_{v,D}\}_{v \in V, D \in \mathcal{S}})$ be an instance of quantum walk search on $\D$, and let $\sigma \in \Delta_V$. We write
        \[\mathsf{S}^{\pm} := \max_{v \in D}\max_{x \in \D} (\mathsf{S}_v)_x^{\pm}, \qquad \mathsf{U}^{\pm} := \max_{e \in E} \max_{x \in \D} (\mathsf{U}_e)_x^{\pm}, \qquad \text{and} \qquad \mathsf{C}^{\pm} := \max_{v \in V}\max_{x \in \D} (\mathsf{C}_{v,D_{v,x}})_x^{\pm},\]
        and we let $\mathsf{M} := \sqrt{\mathsf{M}^+\mathsf{M}^-}$, with $\mathsf{M} \in \{\mathsf{S}, \mathsf{U}, \mathsf{C}\}$. Let $t \geq 1$, and let $R \geq 0$ such that for all $x \in \D_+$, we have $R_{\mathrm{eff}}(P^t, \sigma \leftrightarrow M_x) \leq R$. Then, we can solve the detection version of the quantum walk search problem with a span program $\mathcal{P}$ with complexity
        \[C(\mathcal{P}) \in O\left(\mathsf{S} + \sqrt{tR}\mathsf{U} + \sqrt{1+R}\mathsf{C}\right).\]
    \end{theorem}

    \begin{proof}
        First, we rescale all the hyperedges using \cref{lem:rescale-state-reflection-problem}, such that we can without loss of generality assume that $\mathsf{M}^- = 1$ and $\mathsf{M}^+ = \mathsf{M}^2$, for all $\mathsf{M} \in \{\mathsf{S},\mathsf{U},\mathsf{C}\}$.

        Now, we use \cref{thm:quantum-walk-detection} to construct a generalized graph composition for the quantum walk search problem. We choose $\sigma \in \Delta_V$, and we let $\tau = (\sigma + \pi)/2 \in \Delta_V$\footnote{The idea to take the average between the stationary distribution and the target distribution was already used before in \cite{bencivenga2020sampling}. Here, we rehash that idea for a different purpose: we want the routine to detect a marked vertex if the initial sample already finds it with high probability. Adding $\sigma_v$ in the denominator prevents the effective resistance from blowing up in that case.}. Observe from \cref{thm:pi-unique} that $\tau$ has full support. For all $x \in \D$ for which $M_x \neq \varnothing$, we define $\mu_x = \sigma_x$, and we let $\nu_x$ be the net-flow with support in $M_x$ that minimizes the effective resistance between $\mu_x$ and $\nu_x$ in the graph induced by $P^t$. With \cref{thm:quantum-walk-detection}, we obtain a span program $\mathcal{P}$ that solves the detection version of the quantum walk search problem.

        It remains to compute the witness sizes. For a negative instance $x \in \D^-$, we have
        \[w_-(x,\mathcal{P}) \leq \sum_{v \in \supp(\sigma)} \sigma_v(\mathsf{S}_v)_x^- + \sum_{e \in E} \frac{(\mathsf{U}_e)_x^-}{r_e} + \sum_{v \in V} \tau_v(\mathsf{C}_{v,D_{v,x}})_x^- \leq 1 + 1 + 1 \in O(1),\]
        where we use that all of the negative witness complexities are upper bounded by $1$, that $\sigma$ and $\tau$ are probability distributions, and that the sum of the resistances is at most a constant because it corresponds to a Markov process $M = (V,P)$, as in \cref{thm:random-walk-characterization}.

        For a positive instance, $x \in \D_+$, we upper bound the positive witness complexity as
        \begin{align*}
            w_+(x,\mathcal{P}) &\leq \sum_{v \in supp(\sigma)} \frac{(\mu_x)_v^2}{\sigma_v}(\mathsf{S}_v)_x^+ + R_{\mathrm{eff}}(G, r \circ (\mathsf{U}_e)_x^+; \mu_x - \nu_x) + \sum_{v \in M_x} \frac{(\nu_x)_v^2}{\tau_v}(\mathsf{C}_{v,D_{v,x}})_x^+ \\
            &\leq \sum_{v \in \supp(\sigma)} \sigma_v\mathsf{S}^2 + R_{\mathrm{eff}}(G, r; \mu_x - \nu_x)\mathsf{U}^2 + 2\sum_{v \in M_x} \frac{(\nu_x)_v^2}{\sigma_v + \pi_v}\mathsf{C}^2.
        \end{align*}

        We upper bound the second and third term separately. For the second term, we use \cref{lem:fast-forwarding} to obtain that
        \[R_{\mathrm{eff}}(G, r; \mu_x - \nu_x)\mathsf{U}^2 = R_{\mathrm{eff}}(P; \mu_x - \nu_t)\mathsf{U}^2 \leq t \cdot R_{\mathrm{eff}}(P^t; \mu_x - \nu_x)\mathsf{U}^2 = t \cdot R_{\mathrm{eff}}(P^t; \sigma \leftrightarrow M_x)\mathsf{U}^2 \leq tR\mathsf{U}^2.\]

        For the third term, we define $\nu_x' := \nu_x - \sigma|_{M_x}$. Now, we find that $\sigma - \nu_x = \sigma|_{V \setminus M_x} - \nu_x'$, and moreover
        \[\sum_{v \in M_x} (\nu_x')_v = \sum_{v \in M_x} (\nu_x)_v - \sum_{v \in M_x} \sigma_v = 1 - \sum_{v \in M_x} \sigma_v = \sum_{v \in V \setminus M_x} \sigma_v.\]
        Combining this with \cref{lem:fraction-vs-eff-resistance}, we obtain that
        \begin{align*}
            \sum_{v \in M_x} \frac{(\nu_x)_v^2}{\sigma_v + \pi_v} &= \sum_{v \in M_x} \frac{(\sigma_v + (\nu_x')_v)^2}{\sigma_v + \pi_v} = \sum_{v \in M_x} \frac{\sigma_v(\sigma_v + 2(\nu_x')_v)}{\sigma_v + \pi_v} + \sum_{v \in M_x} \frac{(\nu_x')_v^2}{\sigma_v + \pi_v} \\
            &\leq \sum_{v \in M_x} (\sigma_v + 2(\nu_x')_v) + \sum_{v \in M_x} \frac{(\nu_x')_v^2}{\pi_v} = \sum_{v \in M_x} \sigma_v + 2\sum_{v \in V \setminus M_x} \sigma_v + \sum_{v \in M_x} \frac{(\nu_x')_v^2}{\pi_v} \\
            &\leq 2 + 2R_{\text{eff}}(P^t; \sigma|_{V \setminus M_x} - \nu_x') = 2 + 2R_{\text{eff}}(P^t; \sigma - \nu_x) = 2(1 + R_{\mathrm{eff}}(P^t; \sigma \leftrightarrow M_x)) \leq 2(1+R).
        \end{align*}
        Putting everything together, we obtain that
        \[w_+(x,\mathcal{P}) \in O\left(\mathsf{S}^2 + tR\mathsf{U}^2 + (1+R)\mathsf{C}^2\right),\]
        which using that for all $a,b > 0$, $\sqrt{a+b} \in O(\sqrt{a} + \sqrt{b})$, implies the result.
    \end{proof}

    We remark here that this result relates to \cite[Theorem~1/Theorem~13]{apers2021unified}, as well as \cite[Theorem~2/Theorem~27]{apers2021unified}. The result we present here is weaker, in the sense that it only solves the detection version of the quantum walk search problem, whereas the results from Apers, Gily\'en and Jeffery solve both the detection and the finding version of the problem simultaneously. On the other hand, it is stronger in the sense that it removes all logarithmic factors from the query complexity bound. On top of that, as we will see in the next section, our result can be extended to allow for amortization of the costs too.

    We observe that the complexity statement in our theorem differs slightly from the one in \cite{apers2021unified}, as the prefactor in front of the checking cost, $\mathsf{C}$, is $\sqrt{1+R}$ for us, whereas it is $\sqrt{R}$ in \cite{apers2021unified}. This is because Apers, Gily\'en and Jeffery et al.\ assume that the effective resistance between $\mu_x$ and $\nu_x$ is always at least a constant, i.e., $R \in \Omega(1)$. However, it is possible to construct situations where this effective resistance is much smaller than one, i.e., $R \ll 1$, in which case we still need an additive overhead of $\mathsf{C}$ in the complexity statement.

    Curiously, the construction of \cite{apers2021unified} requires a choice of $t$ in the procedural description of the algorithm. Here, though, the parameter $t$ only shows up in the analysis of the algorithm, but not in its implementation. Thus, there is an interesting qualitative improvement we make here, i.e., we give a single algorithm with a query complexity that beats any choice for the parameter $t$ in the construction from Apers, Gily\'en and Jeffery.

    Finally, we remark that this result can be ported to the finding case as well, in the same settings as considered in the previous section. This also shaves off some of the log-factors in the finding results in \cite{apers2021unified}. However, the construction in \cite{apers2021unified} also allows for solving the finding version of the problem without knowing the fraction of marked vertices beforehand, so we leave it for future research to recover this setting from the graph composition construction as well.

    \subsection{Relation to quantum subroutine composition~\cite{jeffery2022quantum}}

    In this subsection, we compare the results we obtain in this work to those obtained by Jeffery~\cite{jeffery2022quantum}. We start by obtaining a version of the variable-time quantum walk search span program, where all the weight is put on the negative witness complexity.

    \begin{theorem}[Variable-query quantum walk search -- detection version]
        \label{thm:var-time-quantum-walk-detection}
        Suppose we have the instance of quantum walk search $(G, r, x \mapsto M_x, (v,x) \mapsto D_{v,x}, \{S_v\}_{v \in V}, \{U_e\}_{e \in E}, \{C_{v,D}\}_{v \in V, D \in \mathcal{S}})$. Let $M = (V,P)$ be the corresponding random walk on $G$, with stationary distribution $\pi \in \Delta_V$. Let $\sigma,\tau \in \Delta_V$ be distributions with full support. We write, for all $v \in V$, $e \in E$ and $D \in \mathcal{S}$,
        \[\mathsf{S}_v^{\pm} := \max_{x \in \D_{\pm}} (\mathsf{S}_v)_x^{\pm}, \qquad \mathsf{U}_e^{\pm} := \max_{x \in \D_{\pm}} (\mathsf{U}_e)_x^{\pm}, \qquad \text{and} \qquad \mathsf{C}_{v,D}^{\pm} := \max_{x \in \D_{\pm}} (\mathsf{C}_{v,D})_x^{\pm}.\]
        Next, we take two different definitions of several quantities:
        \begin{enumerate}[nosep]
            \item Let $R := \max_{x \in \D_+} R_{\mathrm{eff}}(P; \sigma \leftrightarrow M_x)$, and for all $x \in \D_+$, let $\nu_x \in \Delta_{M_x}$ be the target distribution that minimizes the flow between $\sigma$ and $M_x$, i.e., $R_{\mathrm{eff}}(P; \sigma \leftrightarrow M_x) = R_{\mathrm{eff}}(P; \sigma - \nu_x)$. For all $x \in \D_+$, let $\varepsilon_x := \E_{v \sim \tau}[(\nu_x)_v^2/\tau_v]^{-1}$, and let $\varepsilon := \min_{x \in \D_+} \varepsilon_x$.
            \item For all $x \in \D_+$, let $\varepsilon_x = \P_{v \sim \tau} [v \in M_x]$, and $\nu_x := (\tau|_{M_x})/\varepsilon_x \in \Delta_{M_x}$. Next, let $\varepsilon = \min_{x \in \D_+} \varepsilon_x$, and $R := \max_{x \in \D_+} R_{\mathrm{eff}}(P; \mu_x - \nu_x)$.
        \end{enumerate}
        In both cases, we can solve the detection version of the quantum walk search problem with a span program $\mathcal{P}$, with complexities
        \begin{align*}
            w_+(x,\mathcal{P}) &\in O\left(\frac{\underset{v \sim \sigma}{\E}[(\mathsf{S}_v)_x^+]}{\underset{v \sim \sigma}{\E}[\mathsf{S}_v^+]} + \frac{R_{\mathrm{eff}}\left(G, r \circ \frac{(\mathsf{U}_{\cdot})_x^+}{\mathsf{U}_{\cdot}^+}; \sigma - \nu_x\right)}{R} + \underset{v \sim \tau}{\E} \left[\frac{(\nu_x)_v^2}{\tau_v} \cdot \frac{\varepsilon(\mathsf{C}_{v,D_{v,x}})_x^+}{\mathsf{C}_{v,D_{v,x}}^+}\right] \right) \subseteq O(1), \\
            w_-(x,\mathcal{P}) &\in O\left(\underset{v \sim \sigma}{\E}\left[\mathsf{S}_v^+\right] \cdot \underset{v \sim \sigma}{\E} \left[(\mathsf{S}_v)_x^-\right] + R \cdot \underset{\substack{v \sim \pi \\ w \sim P_{v\cdot}}}{\E} \left[\mathsf{U}^+_{vw}(\mathsf{U}_{vw})_x^-\right] + \frac{1}{\varepsilon} \cdot \underset{v \sim \tau}{\E} \left[\mathsf{C}_{v,D_{v,x}}^+(\mathsf{C}_{v,D_{v,x}})_x^-\right]\right),
        \end{align*}
        for $x \in \D_+$ and $x \in \D_-$, respectively.
    \end{theorem}

    \begin{proof}
        From \cref{lem:rescale-state-reflection-problem}, we can rescale the setup routines so that for every $v \in V$, $(\mathsf{S}'_v)_x^{\pm} = (\mathsf{S}_v)_x^{\pm} \cdot \E_{v \sim \sigma}[\mathsf{S}_v^+]^{\mp1}$. Similarly for all $e \in E$, we rescale $U_e$ such that $(\mathsf{U}_e')_x^{\pm} = (\mathsf{U}_e)_x^{\pm} \cdot (R\mathsf{U}_e^+)^{\mp1}$, and for all $v \in V$ and $D \in \mathcal{S}$, we rescale $C_{v,D}$ such that $(\mathsf{C}_{v,D}')_x^{\pm} = (\mathsf{C}_{v,D})_x^{\pm} \cdot (\mathsf{C}_{v,D_{v,x}}^+/\varepsilon)^{\mp1}$. Then, the resulting complexities follow from \cref{thm:quantum-walk-detection}, with the rescaled setup, update and checking routines.
    \end{proof}

    Using the above theorem, we can recover the results from \cite[Corollary~4.2]{jeffery2022quantum}. The idea is that if we have quantum Las Vegas algorithms for the setup, update, and checking routines, then we can convert them into hyperedges using \cref{thm:las-vegas-algorithm,thm:hyperedge-from-database-function}. If we know the expected number of queries beforehand, we simply use $\alpha_t = 1$ for all $t \in \N_0$ in \cref{thm:las-vegas-algorithm}, which yields $\mathsf{W}_x^{\pm} = \E[T_x]$ for every Las Vegas routine with stopping time $T_x$ on input $x \in \D$. Since we now know all these numbers $\mathsf{W}_x^{\pm}$, we can use them in the construction in \cref{thm:var-time-quantum-walk-detection}. This yields all the complexities from \cite[Corollary~4.2]{jeffery2022quantum} for known query complexities.

    On the other hand, if we only know a crude upper bound $T$ on their query complexities, we use $\alpha_t = 1/(t+1)$, for all $t \in \N_0$ in \cref{thm:las-vegas-algorithm}. We obtain $\mathsf{W}_x^+ \in O(\log(T))$, and $\mathsf{W}_x^- \in O(\E[T_x^2])$, for every Las Vegas routine with stopping time $T_x$ on input $x \in \D$. Since we only need to know upper bounds on $\mathsf{W}_x^+$ in \cref{thm:var-time-quantum-walk-detection}, this suffices for our construction. This yields all the complexities from \cite[Corollary~4.2]{jeffery2022quantum} for unknown query complexities.

    Note that we improve by a factor of $\log(T)$ over \cite[Corollary~4.2]{jeffery2022quantum} with our approach. The reason is because Jeffery also analyzes the time complexity of the construction, which requires an extra $\log(T)$-overhead. We show here that this overhead is not required if one only cares about queries.

    We can now combine the previous theorem with \cref{lem:fast-forwarding,lem:fraction-vs-eff-resistance}, to obtain a variable-query result for the MNRS-framework. This improves over \cite[Theorem~4.3]{jeffery2022quantum}.

    \begin{corollary}[Variable-query MNRS]
        \label{thm:mnrs}
        Suppose we have the instance of quantum walk search $(G, r, x \mapsto M_x, (v,x) \mapsto D_{v,x}, \{S_v\}_{v \in V}, \{U_e\}_{e \in E}, \{C_{v,D}\}_{v \in V, D \in \mathcal{S}})$. Let $\pi$ be the stationary distribution of the random walk $P$ on $G$, and let $\delta$ be its spectral gap. For all $x \in \D_+$, let $\varepsilon_x := \P_{v \sim \pi}[v \in M_x]$, and we write $\varepsilon := \min_{x \in \D_+} \varepsilon_x$. Then, we can solve the detection version of the quantum walk search problem with a span program $\mathcal{P}$, whose complexities satisfy
        \begin{align*}
            w_+(x,\mathcal{P}) &\in O(1), \\
            w_-(x,\mathcal{P}) &\in O\left(\underset{v \sim \pi}{\E}\left[\mathsf{S}_v^+\right] \cdot \underset{v \sim \pi}{\E} \left[(\mathsf{S}_v)_x^-\right] + \frac{1}{\varepsilon} \left(\frac{1}{\delta} \cdot \underset{\substack{v \sim \pi \\ w \sim P_{v\cdot}}}{\E} \left[\mathsf{U}^+_{vw}(\mathsf{U}_{vw})_x^-\right] + \underset{v \sim \pi}{\E} \left[\mathsf{C}_{v,D_{v,x}}^+(\mathsf{C}_{v,D_{v,x}})_x^-\right]\right)\right),
        \end{align*}
        for $x \in \D_+$ and $x \in \D_-$, respectively.
    \end{corollary}

    \begin{proof}
        We use the second claim from \cref{thm:var-time-quantum-walk-detection}, with $\sigma = \tau = \pi$. It then remains to prove that
        \[R_{\mathrm{eff}}\left(G,r;\pi - \frac{\pi|_{M_x}}{\varepsilon_x}\right) \leq \frac{1}{\delta\varepsilon_x}.\]
        To that end, we write $\xi := \pi - \pi|_{M_x}/\varepsilon_x$ and $D = \diag(\pi)$. We use the second claim from \cref{lem:fast-forwarding} and observe that
        \begin{align*}
            R_{\mathrm{eff}}\left(G, r; \pi - \frac{\pi|_{M_x}}{\varepsilon_x}\right) &\leq \frac{\norm{D^{-1/2}\xi}_2^2}{\delta} = \frac{1}{\delta} \sum_{v \in V} \frac{\xi_v^2}{\pi_v} = \frac{1}{\delta} \left[\sum_{v \in V \setminus M_x} \pi_v + \sum_{v \in M_x} \left(\frac{1}{\varepsilon_x} - 1\right)^2\pi_v\right] \\
            &= \frac{1}{\delta}\left[1 - \frac{2}{\varepsilon_x}\sum_{v \in V \setminus M_x} \pi_v + \frac{1}{\varepsilon_x^2}\sum_{v \in M_x} \pi_v\right] = \frac{1}{\delta}\left[\frac{1}{\varepsilon_x} - 1\right] \leq \frac{1}{\delta\varepsilon_x}.\qedhere
        \end{align*}
    \end{proof}

    The above analysis improves over \cite{jeffery2022quantum} by getting rid of the $\sqrt{\log(1/\pi_{\min})}$ factor. This is significant, because one of the most prominent usecases of the MNRS-framework is implementing quantum walk search over a Johnson graph $J(n,r)$. Since the number of vertices of such a Johnson graph is $\binom{n}{r}$, and the stationary distribution is uniform, we obtain in this setting that $\log(1/\pi_{\min}) = \log\binom{n}{r} \sim r\log(n)$, which introduces polynomial overhead in the cases where we choose $r$ polynomial in $n$. Jeffery already conjectured just before \cite[Corollary~4.3]{jeffery2022quantum} that this overhead is unnecessary, and we confirm that conjecture here.

    Furthermore, the previous theorem is a generalization of \cite[Theorem~3]{carette2020extended}. There, the authors provide a theorem that amortizes the costs of the setup, update, and checking routines, in the same way as in \cref{thm:mnrs}. However, their result only works when the quantum walk is performed on a Johnson graph, and the objective routines are given by learning graphs. We generalize to quantum walks on arbitrary graphs, and with the input routines being represented as feasible solutions to the adversary bound, rather than learning graphs.

    Finally, we remark that we can port these results in exactly the same way as in the general case to the finding case, i.e., if we are promised that the marked vertex is unique, or if we know the fraction of marked vertices a priori. To the best of our knowledge, this is the first result of its kind that is able to solve the finding version of the quantum random walk search question on a general graph, while simultaneously allowing for amortization of its subroutines. We leave it for future work to figure out if one can also solve the finding question in the more generic sense, i.e., without knowing the fraction of marked vertices beforehand.

    \section*{Acknowledgements}

    I would like to thank Simon Apers, Zeph Landau, Galina Pass, Stacey Jeffery and Sebastian Zur for fruitful discussions. I am supported by a Simons-CIQC postdoctoral fellowship through NSF QLCI Grant No. 2016245.

    \bibliographystyle{alpha}
    \bibliography{references}
\end{document}